\documentclass[a4paper]{article}
\usepackage{amsmath,amsfonts,amsthm}
\usepackage{bm}
\usepackage{natbib}
\usepackage{graphicx}
\usepackage{fullpage}
\usepackage{hyperref}

\usepackage[plain,noend]{algorithm2e}

\usepackage{titlesec}

\titleformat*{\section}{\bfseries}
\titleformat*{\subsection}{\bfseries}
\titleformat*{\subsubsection}{\bfseries}
\titleformat*{\paragraph}{\bfseries}
\titleformat*{\subparagraph}{\bfseries}

\def\T{{ \mathrm{\scriptscriptstyle T} }}

\def\T{{ \mathrm{\scriptscriptstyle T} }}

\newcommand{\eff}{\operatorname{eff}}
\newcommand{\E}{E}

\newtheorem{assumption}{Assumption}
\newtheorem{definition}{Definition}
\newtheorem{theorem}{Theorem}
\newtheorem{lemma}{Lemma}
\newtheorem{algo}{Algorithm}
\newtheorem{proposition}{Proposition}
\title{ \Large Designs for generalized linear models with\\ random block effects via information matrix approximations}

\author{ \normalsize T. W. Waite\footnote{t.w.waite@southampton.ac.uk} \, and D. C. Woods\footnote{d.woods@southampton.ac.uk}
\\  \normalsize Southampton Statistical Sciences Research Institute,\\  \normalsize University of Southampton, SO17 1BJ, U.K. }
\date{}
\begin{document}

% Left and right page headers
 \markboth{T. W. Waite and D. C. Woods}{Designs for generalized linear models with random block effects}

 \maketitle
 
 \begin{abstract}
%For an experiment measuring independent discrete responses, a generalized linear model, such as the logistic or log-linear, is typically used to analyse the data. In blocked experiments, where observations from the same block are potentially correlated, it may be appropriate to include random effects in the predictor, thus producing a generalized linear mixed model. Selecting optimal designs for such models is complicated by the fact that the Fisher information matrix, on which most optimality criteria are based, is computationally expensive to evaluate. In addition, the dependence of the information matrix on the unknown values of the parameters must be overcome by, for example, use of a pseudo-Bayesian approach.  We use closed-form approximations, derived from marginal quasi-likelihood, to develop computationally inexpensive surrogates for the information matrix to obtain $D$-optimal designs. This approach reduces the computational burden substantially, enabling straightforward selection of multi-factor designs. The accuracy of the closed-form approximations is explored for the first time using a novel computational approximation. It is found that correcting for the marginal attenuation of parameters in binary-response models yields much improved designs. 

%For a blocked experiment measuring discrete responses, a generalized linear mixed model, including random block effects, incorporates potential correlations between observations from the same block. 

The selection of optimal designs for generalized linear mixed models is complicated by the fact that the Fisher information matrix, on which most optimality criteria depend, is computationally expensive to evaluate. Our focus is on the design of experiments for likelihood estimation of parameters in the conditional model. We provide two novel approximations that substantially reduce the computational cost of evaluating the information matrix by complete enumeration of response outcomes, or Monte Carlo approximations thereof: (i) an asymptotic approximation which is accurate when there is strong dependence between observations in the same block; (ii) an approximation via Kriging interpolators. For logistic random intercept models, we show how interpolation can be especially effective for finding pseudo-Bayesian designs that incorporate uncertainty in the values of the model parameters. The new results are used to provide the first evaluation of the efficiency, for estimating conditional models, of optimal designs from closed-form approximations to the information matrix derived from marginal models. It is found that correcting for the marginal attenuation of parameters in binary-response models yields much improved designs, typically with very high efficiencies. However, in some experiments exhibiting strong dependence, designs for marginal models may still be inefficient for conditional modelling. Our asymptotic results provide some theoretical insights into why such inefficiencies occur.\\[1ex]

\end{abstract}

%%%%%%%%%%%%%%%%%%%%%%%%%%%%%%%%%%%%%%%%%%%%%%%%%%%%%%%%%%%%%%%%%%%%%%%%%%%%%%%%%%%%%
%%%%%%%%%%%%%%%%%%%%%%%%%%%%%%%%%%%%%%%%%%%%%%%%%%%%%%%%%%%%%%%%%%%%%%%%%%%%%%%%%%%%%
\section{Introduction}

There is increasing recognition of the need to design experiments in situations where a linear model with only fixed effects cannot adequately capture the essential features of the data. In particular, there is a growing body of work on optimal design for generalized linear models \citep[for example,{}][]{chaloner-larntz,woods06,yang11} which can be used when the response variable follows a non-normal distribution from the exponential family. An even more substantial literature addresses the problem of optimal design when there is heterogeneity between blocks in an experiment, using a linear mixed model for normally distributed responses with random block effects (for example, \citealp{cheng} and \citealp{goos-vdb}). In many practical contexts, such as the industrial experiment described by \citet{woods-vdv}, both of these features (non-normality and heterogeneity) are present. For such experiments, particularly where the response is discrete, for example~binary or count data, generalized linear mixed models may be an appropriate modelling choice. In this paper, we develop and compare optimal design methodologies for this family of models.

%Another popular approach is to fit marginal models using generalized estimating equations \citep{liang-zeger}. This semiparametric method makes use of a `working correlation' matrix, that may be misspecified while still retaining estimation consistency. The optimal design approach of \citet{woods-vdv} assumes the true and working correlation matrices are equal, with an exchangeable, nearest neighbour or autoregressive structure. However, in general, for binary response such a structure may not be possible, because usually the correlation depends on the treatments \citep{chaganty}, and so a priori  there is not necessarily any probability model under which the designs of \citet{woods-vdv} are efficient. A theoretical advantage of the conditional model design paradigm adopted in this paper is that it is clear under which mathematically feasible circumstances the design is optimal or efficient. In Section \ref{sec:amql}, we give an extension of the methods of \citet{woods-vdv} that often produces designs that are efficient for conditional, rather than marginal, modelling.

We find $D$-optimal designs, that is, designs that minimize the volume of the asymptotic confidence ellipsoid for the model parameters by maximizing the determinant of the Fisher information matrix. Dependence of an optimal design on the unknown values of the parameters is addressed using a pseudo-Bayesian approach. The key technical difficulty in the construction of $D$-optimal designs for generalized linear mixed models is that the information matrix is computationally expensive to evaluate. 

The adoption of a mixed model implies a marginal distribution for the response with intra-block correlation: two responses from units within the same block are correlated and responses from units in different blocks are uncorrelated. Existing approaches to optimal design for correlated discrete responses (see, for example, \citealp{moerbeek05}, \citealp{niaparast} and \citealp{woods-vdv}) are tailored for inferential methods, such as quasi-likelihood and generalized estimating equations, that use only the first and second order moments of the marginal distribution, and approximations thereof, for parameter estimation. Our focus is on design for direct (likelihood) estimation of parameters in the conditional model for which we present novel asymptotic and computational approximations to the information matrix. For binary data, we also adapt and extend marginal approximations to provide more efficient designs for the conditional model. We compare designs from the various approximations to those found from computationally expensive ``gold standard'' approximations throughout our examples, using either na\"ive outcome-enumeration or Monte Carlo methods (Section~\ref{sec:complenum}).

\section{Preliminaries}
\subsection{Generalized linear mixed models for blocked experiments}
\label{sec:model}
We denote the response for the $j$th unit in the $i$th block by $y_{ij}$, and the corresponding treatment vector of values taken by the $q$ controllable variables by $x_{ij} \in \mathcal{X}\subseteq\mathbb{R}^q$ ($i=1,\ldots,n$, $j=1,\ldots,m_i$). Also let $\zeta_i=(x_{i1},\ldots,x_{im_i})\in \mathcal{X}^{m_i}$ denote the $m_i$ treatment vectors in the $i$th block. Then, for a generalized linear mixed model, there is a vector, $u_i$, of $r$ random effects associated with the $i$th block in the experiment. Conditional on $u_i$, the responses in block $i$ are independent and follow an exponential family distribution, $y_{ij} | u_i \sim \pi(x_{ij};u_i,\beta)$, with mean $\mu_{ij}=\mu(x_{ij};u_i,\beta)$ and variance $\varphi v(x_{ij};u_i, \beta)$. For the models we consider, the dispersion parameter $\varphi=1$.  The mean function $\mu(x;u,\beta)$ is defined by
\begin{alignat}{2}
g\{\mu(x;u,\beta)\} &= \nu (x;u,\beta) \,, &\hspace{0.5cm}
\nu(x;u,\beta) &= f^{\T}(x) \beta + z^\T(x) u\,,
\label{eq:link}
\end{alignat}
\noindent where $f : \mathcal{X} \to \mathbb{R}^p$ and $z : \mathcal{X} \to \mathbb{R}^{r}$ are known vectors of regressor functions, with $z$ typically a subvector of $f$, and $\beta$ is a $p$-vector of fixed regression parameters. The link function $g$ relates the linear predictor $\nu$ to the mean response. Denote by $h$ the inverse link function, $g^{-1}$. To fully determine the model, assumptions about the distribution of $u_i$ are necessary; we specify independent $u_i \sim \mbox{\textsc{mvn}}(0, G)$, with $G$ an arbitrary covariance matrix. The presence of random effects in the linear predictor introduces a correlation between observations from experimental units in the same block. In this paper, we focus principally on random intercept models appropriate for blocked experiments, where $r=1$, $G=(\sigma^2)$, $\sigma^2>0$, and $z(x)=(1)$.
%This is natural since in experimental design one often assumes additive block effects. 

We assume for simplicity that $m_i=m$, and say treatment blocks $\zeta_1,\zeta_2\in\mathcal{X}^m$ are \emph{equivalent} if one can be obtained from the other by treatment permutation. Without loss of generality, we may assume that the blocks are ordered so that the design is supported on $\zeta_1,\ldots,\zeta_b$ ($1\leq b \leq n$), and that no pair from among the first $b$ blocks is equivalent. 
%Without loss of generality, we may assume that the design is supported on $\zeta_1,\ldots,\zeta_b$ ($1\leq b \leq n$) and that these blocks are not equivalent, in the sense that one block cannot be obtained from another by permuting treatments. 
Let $w_k$, $ k =1,\ldots, b$, be the proportion of blocks equivalent to the $k$th support block $\zeta_k$, then we have the following concise notation for a design:
\begin{equation}
\xi = \left\{
\begin{array}{ccc}
\zeta_1, &\cdots, &\zeta_b \\
w_1, &\cdots, & w_b
\end{array}
\right\}\,,
\label{eq:general-design}
\end{equation}
\noindent where $0 < w_k \leq 1$ and $\sum_{k=1}^{b}w_k =1$. As defined above, $nw_k$ is a positive integer. We focus on approximate block designs, which relax this constraint \citep[see also][]{cheng}. Note that we do not impose any restrictions on the form of $\zeta_k$, and that designs may differ in the $w_k$, the $\zeta_k$ and the value of $b$. We restrict to designs with finite support, $b\le n<\infty$.

%Note that we allow different designs to have different $\zeta_k$, and we do not impose any restrictions on the form of $\zeta_k$, but we may bound $b$ (clearly $b\leq n$, but see also Section~\ref{sec:opt-crit}).  %Clearly in order to implement a general approximate design some rounding of the $w_k$ will be required, see \citet{pukelsheim1992efficient}.

\subsection{Information matrix}
\label{sec:complenum}

Let $\theta$ denote the complete vector of parameters for model~(\ref{eq:link}). Thus $\theta$ includes the fixed effects parameters $\beta$ as well as any parameters specifying the distribution of $u_i$. Denote by $M_{\beta}$ the information matrix for $\beta$, holding all other components of $\theta$ fixed. The use of $M_{\beta}$ is appropriate for assessing the precision of a maximum likelihood estimator $\hat{\beta}$ assuming known variance components. In common with many papers on design for both linear mixed models \citep{cheng, goos-vdb} and specific examples of their generalized counterparts \citep{moerbeek05,tekle,niaparast-schwabe}, we do not consider the additional variability in $\hat{\beta}$ introduced when the variance components also require estimation.

For the approximate block design $\xi$ in~(\ref{eq:general-design}), the information matrix $M_{\beta}$ depends on $\theta$ and, as observations in different blocks are independent, can be decomposed into a weighted sum of the information matrices for each support block,
\begin{equation}
 M_{\beta}(\xi,\theta) =  \sum_{k=1}^{b} w_k M_{\beta}(\zeta_k,\theta) \,.
\label{eq:infdesign}
\end{equation}

%Using this additivity property, Caratheodory's theorem shows that for local optimality criteria, i.e. where the objective function is a function of  $M_\beta(\xi,\theta)$ and a particular value of $\theta$ is assumed, there is always an optimal design supported on at most $p(p+1)/2 +1$ distinct blocks.

The information matrix for an arbitrary block $\zeta=(x_1 ,\ldots, x_m) \in \mathcal{X}^m$ is
\begin{equation} M_{\beta}(\zeta,\theta) = F^\T E_{Y}    
					\left\{ 	
							P(Y|\theta, \zeta)^{-2}
								\left(
								 \frac{\partial P(Y|\theta, \zeta)}
								 	{\partial \eta}
								\right)
								\left(
								 \frac{ \partial P(Y|\theta, \zeta)}
								 	{\partial \eta}
								\right)^\T
				         \right\} F \,,
\label{eq:info}
\end{equation}
\noindent where $Y=(y_1,\ldots, y_m)^\T$ denotes the response vector or outcome corresponding to $\zeta$, $P(Y|\theta,\zeta)$ is the marginal likelihood of the model parameters, $\eta = (f^\T(x_1)\beta , \ldots, f^\T(x_m)\beta)^\T$, and $F=[f(x_1), \ldots, f(x_m)]^\T$ is the model matrix. The likelihood and its derivative are of the form
\begin{equation}
P(Y|\theta,\zeta) = \int_{\mathbb{R}^r} 
						P(Y|u,\theta,\zeta) 
							f_u(u) du \,,
	\quad
	\frac{\partial P(Y|\theta, \zeta)}
	{\partial \eta}
	= \int_{\mathbb{R}^r} 
		\frac{\partial P(Y|u,\theta,\zeta)}
			{\partial \eta } 
			f_u(u) du\,,
\label{eq:marglik}
\end{equation}
where $P(Y|u,\theta,\zeta)$ is the (exponential family) conditional probability density of $Y$ given $u$ and $f_u$ is the density function of an $\textsc{mvn}(0,G)$ random variable. Typically a closed form for the partial derivative of the conditional density is available. For random intercept models, the integrals in~\eqref{eq:marglik} can be evaluated numerically using Gauss-Hermite quadrature. %However, typically the expectation in~\eqref{eq:info} is harder to compute.  

For models with binary response, the expectation in~\eqref{eq:info} can be evaluated by enumeration of outcomes $Y\in \{0,1\}^m$. Expanding the expectation, we obtain
\begin{equation}
M_{\beta}(\zeta,\theta) =  F^\T\sum_{Y \in \{0,1\}^m} 
				 P(Y|\theta, \zeta)^{-1}
								\left(
								 \frac{\partial P(Y|\theta, \zeta)}
								 	{\partial \eta}
								\right)
								\left(
								 \frac{ \partial P(Y|\theta, \zeta)}
								 	{\partial \eta}
								\right)^\T  F\,,
\label{eq:blockinf}
\end{equation}
\noindent where the sum is over all possible response patterns in block $\zeta$. An obvious approximation to information matrix~\eqref{eq:info} is via~\eqref{eq:blockinf} with numerical approximation of~\eqref{eq:marglik} using quadrature. We call this approach \textit{na\"ive outcome-enumeration}. Clearly, for even moderately sized blocks, such an approximation will be computationally expensive.

%\begin{definition}[numerical outcome enumeration]
%\end{definition}

For other response distributions, such as Poisson, the expectation in \eqref{eq:info} can be approximated, in principle, by Monte Carlo sampling of response vectors $Y$. In practice, to obtain reasonable precision in the approximation of the information matrix using this method, it is necessary to consider many more than the $2^m$ possible distinct outcomes obtained from a binary model. %Hence, the general problem is more difficult for a non-binary response, except for the straightforward normal-theory linear mixed model.

\subsection{Optimality criteria}
\label{sec:opt-crit}
We study both locally $D$-optimal designs, i.e. $\xi^{\ast}_D = \operatorname*{arg\,max}_{\xi} |M_\beta(\xi , \theta)|$ for an assumed value of $\theta$, and \mbox{(pseudo-)}Bayesian designs. From~\eqref{eq:infdesign} and an application of Caratheodory's theorem (e.g. \citealp[][p.16]{silvey}), it follows that there is always a locally $D$-optimal design supported on at most $p(p+1)/2 +1$ distinct blocks. The pseudo-Bayesian approach may be used to construct a design that is more robust to misspecification of the model parameters, and requires specification of a prior distribution, $\Lambda$, for $\theta$. Given $\Lambda$, $\xi$ is Bayesian $D$-optimal if it maximizes 
$
\psi(\xi) =E_\theta\{\log  |M_\beta (\xi,\theta)|\}
$ \citep{chaloner-larntz}.
 We do not assume that the resulting analysis will be Bayesian, or that it will use prior distribution $\Lambda$. Care must be taken when the prior distribution has unbounded  support; see \citet{waite2013integrability}.

%%%%%%%%%%%%%%%%%%%%%%%%%%%%%%%%%%%%%%%%%%%%%%%%%%%%%%%%%%%%%%%%%%%%%%%%%%%%%%%%%%%%%
%%%%%%%%%%%%%%%%%%%%%%%%%%%%%%%%%%%%%%%%%%%%%%%%%%%%%%%%%%%%%%%%%%%%%%%%%%%%%%%%%%%%%

\section{Approximations via marginal models}
\subsection{Marginal quasi-likelihood}
\label{sec:mql}

\citet{bres-clay} discussed marginal quasi-likelihood as a computationally inexpensive, approximate method for estimating the parameters of a generalized linear mixed model. The method is indirect in that it applies standard quasi-likelihood equations for dependent data \citep[][Sec. 9$\cdot$3]{mccullagh-nelder} to a linearization of the model about the mean value of the random effects. An information matrix approximation corresponding to this method is
\[
M_\beta^{\text{marg}}(\xi,\theta) =  \sum_{k=1}^{b}
            	w_k \,
			F^\T_kV_k^{-1}F_k \,,
	\label{eq:approx-info-matrix}
\]
where $F_k$ is the model matrix for $\zeta_k$, $V_k=\mathcal{V}(\zeta_k,\theta)$ is determined from $\mathcal{V}(\zeta,\theta) = W(\zeta, \theta)^{-1} + Z(\zeta) G Z(\zeta)^\T$, $W(\zeta,\theta)$ is the diagonal matrix with entries $v(x_{1}; 0,\beta),\ldots,v(x_{m} ; 0,\beta)$, and $Z(\zeta)= [z(x_{1}),	\ldots, z(x_{m})	]^\T$. For design using similar methods, see \citet{moerbeek05}.

%For binary responses with logit link, $\pi$ is Bernoulli and $v(x;u,\beta) =\exp \{ f^\T(x)\beta \}/ \{1+\exp[ f^\T(x)\beta] \} ^2\}$; for a Poisson response with log link, $v(x;u,\beta) = \exp \{ f^\T(x)\beta \}$. 

There are several higher-order marginal quasi-likelihood approximations in the literature, for example \citet{gold-ras}. An approximation to the information matrix using a second order method was derived in a 2012 University of Southampton PhD thesis by T.~W. Waite. Use of this approximation does not result in better designs, so we omit the results here. The marginal quasi-likelihood approximation is similar to the first-order approximations used in the design of pharmacokinetic studies (see, for example, \citealp{retout2003further}).

\subsection{Generalized estimating equations}
\label{sec:gee}

Generalized estimating equations \citep{liang-zeger} may be used to estimate parameters when the marginal distribution of the response follows a generalized linear model, making use of a `working correlation' matrix that need not be equal to the true correlation matrix. Typically, a standard structure is used for the working correlation, such as exchangeable, autoregressive or nearest neighbour. However these assumptions are incompatible with most known probability models for dependent discrete responses, in which the correlation is a nontrivial function of the treatments and parameters. Indeed there may not exist any probability model achieving these simple correlation structures with the required univariate marginal distributions if, for example, the working correlation violates the bounds on correlation for binary data \citep[Ch.7]{joe}. Nonetheless, the estimators retain consistency under misspecification of the correlation structure and may still be highly efficient \citep{chaganty}. Note that here we use generalized estimating equations only to obtain an approximation to the mixed model information matrix.

\citet{woods-vdv} found designs for marginal generalized linear models that are $D$-optimal for the generalized estimating equation method under the assumption that the true correlation structure corresponds to a specified working correlation structure. They also found that the resulting designs were robust to a general class of departures from this correlation assumption. Denote the parameters of the assumed marginal model by $\beta^\ast$, the correlation parameter by $\rho$, and assume the marginal model has the same link and variance functions as the conditional model. Then for exchangeable correlation, the inverse asymptotic covariance matrix is
\[
M_\beta^{\text{gen}}(\xi,\beta^\ast, \rho) = \sum_{k=1}^{b} w_k  F_k^\T D_k\{ (V^\ast_k)^{1/2} R(\rho) (V^\ast_k)^{1/2}\}^{-1} D_k F_k \,,
\]
where $D_k$ is the diagonal matrix with entries $1/g'(\mu^\ast_{k1}),\ldots, 1/g'(\mu^\ast_{km})$, $\mu^\ast_{kj}= h\{ f^\T(x_{kj})\beta^\ast \}$, $V^\ast_k$ is the diagonal matrix with entries $v(x_{k1};0, \beta^\ast), \ldots, v(x_{km};0,\beta^\ast)$, and $R(\rho) = (1-\rho)I_m + \rho 1_m 1_m^\T$, with $I_m$ the $m\times m$ identity matrix and $1_m$ an $m$-vector of ones. 

%In general, for binary response, the assumed simple forms for the working correlation may not arise under any probability model because the true correlation depends on the treatments, and is subject to bounds \citep{chaganty}. Thus it is not clear a priori that there is necessarily any probability model under which the designs of \citet{woods-vdv} are efficient. In Section \ref{sec:amql}, we give an extension of these methods that often produces designs that are efficient for estimating a generalized linear mixed model, as seen in Section \ref{sec:bin}.
%
\subsection{Binary response: adjustment for attenuation of parameters}
\label{sec:amql}

Use of marginal quasi-likelihood for the logistic random effects model results in the assumption that the marginal mean has the form $E(y_{ij})\approx g^{-1} \{  f^\T(x_{ij}) \beta \}$ \citep{bres-clay}. \citet{ZLA-GEE} showed that a better approximation to the marginal mean is given by a logistic relationship with attenuated coefficients,
\begin{alignat}{2}
E( y_{ij} ) \approx 
		g^{-1}
				\left\{
					f^\T(x_{ij})\beta/\sqrt{1+c^2 z(x_{ij})^\T G z(x_{ij}) }
					\right\}\,,
		%& \hspace{0.5cm} \beta_\text{att} = \beta\, (1+c^2 \sigma^2)^{-1/2} \,,
\label{eq:general-att}
\end{alignat}
\noindent where $c=16\sqrt{3}/(15\pi)$. For random intercept models this reduces to 
\begin{alignat}{2}
E(y_{ij}) \approx g^{-1}\{ f^\T(x_{ij}) \beta_\text{att}\} \,, & \quad \beta_\text{att} = \beta\, (1+c^2 \sigma^2)^{-1/2} \,.
\label{eq:marg-approx}
\end{alignat}
This suggests that for the logistic random intercept model, more efficient designs might be obtained by adjusting the parameter values to better approximate the marginal mean using~(\ref{eq:marg-approx}). Explicitly, we define the \textit{adjusted marginal quasi-likelihood} information matrix by
\begin{alignat*}{2}
M_\beta^{\text{adj}}(\xi, \theta) = M_\beta^{\text{marg}}(\xi, \theta_{\text{adj}}) \,,& \hspace{0.5cm}
\theta_{\text{adj}} = \left(
						\beta_\text{att}^\T
							, 
						\sigma^2
					\right)^\T \,,
\end{alignat*}
\noindent where $M_\beta^\text{marg}$ is the information matrix for $\beta$ under marginal quasi-likelihood. In models other than the random intercept the attenuation factor depends on $x$, so a constant adjustment cannot be applied for every design point. However, one possibility for a similar approximation may be to apply quasi-likelihood or generalized estimating equations using~\eqref{eq:general-att} as the marginal mean.

%One might be concerned that by adjusting the parameters in this way, we gain accuracy in the marginal mean approximation at the cost of corrupting the approximation to the variance. However, it seems that in practice approximating the marginal mean accurately is more important: in Section \ref{sec:EX1}, the performance of Bayesian designs computed using this adjustment is virtually indistinguishable from designs obtained using the computational approximation outlined in Section~\ref{sec:mlni}.

To extend the methods of \citet{woods-vdv}, we also take account of parameter attenuation by forming the \textit{adjusted generalized estimating equation} approximation,
\begin{equation}\label{eq:GEE}
M_\beta^{\text{adj. gen.}}(\xi, \theta, \rho) = M_\beta^{\text{gen.}}(\xi, \beta_\text{att}, \rho) \,.
\end{equation}
Here we either choose a value of $\rho$ following the guidelines laid out, for estimation, by \citet{chaganty}, or treat $\rho$ as a tuning parameter, i.e.~we choose the value of $\rho$ such that the corresponding $D$-optimal design using~\eqref{eq:GEE} maximizes $|M_\beta|$ approximated via na\"ive outcome-enumeration.

%In our experience, taking account of parameter attenuation as outlined above is critical for obtaining efficient designs for binary-response generalized linear mixed models, with both the marginal quasi-likelihood and generalized estimating equation approximations. 

\section{Theoretical and computational direct approximations for the logistic\\ random intercept model}
\subsection{Asymptotic outcome-enumeration}% for logistic random intercept model}
\label{sec:strongdep}

For the logistic random intercept model, the important case of large $\sigma^2$ results in substantial block-to-block variability and poses a more difficult design problem. In this case, responses in the same block are strongly dependent, and the adjusted marginal quasi-likelihood and adjusted generalized estimating equation designs may perform quite poorly (see Section~\ref{sec:ex-local}). Moreover, for large $\sigma^2$ the na\"ive outcome-enumeration approximation becomes even more computationally expensive, as more quadrature points are required to maintain accuracy in the approximation of the integrals in~\eqref{eq:marglik}. In this section, we develop asymptotic, $\sigma^2\to \infty$, expressions for $P(Y|\theta,\zeta)$ and its derivatives which are combined with~\eqref{eq:blockinf} to provide a new, direct approximation to the information matrix for large finite $\sigma^2$. The additional approximation enables selection of efficient designs for large $\sigma^2$ at low computational cost compared to na\"ive outcome-enumeration (Section~\ref{sec:ex-local}). Our main results are in Theorems~\ref{thm:asymp-prob-quasi-incr}--\ref{thm:quasi-only}; first we define some necessary assumptions and notation. 
%To gain theoretical insight into this phenomenon, we study the asymptotic case, $\sigma^2\to \infty$. 

For fixed values of the conditional parameters, the `marginal effects' in $\beta_\text{att}$ attenuate to zero as $\sigma^2\to\infty$. In order to approximate the more interesting and realistic case where both $\sigma^2$ is large and there are non-zero marginal effects, we assume the following asymptotic conditions.

\begin{assumption} $\beta_\text{att}= \beta/ \sqrt{1+c^2 \sigma^2}$ is fixed.
\end{assumption}
\begin{assumption}
 For each $j$, either $\eta^\ast_j= f^\T(x_j)\beta_\text{att}$ is fixed or there exists $l\neq j$ such that $\eta^\ast_l$ is fixed and $\eta^\ast_l-\eta^\ast_j =o(\sigma^{-1})$.
 \end{assumption}
  
 In order to meet these conditions, we allow the $x_j$ to vary with $\sigma^2$. A simple asymptotic approximation to the information matrix could be derived by treating all $\eta^\ast_j$ as distinct and fixed as $\sigma^2\to \infty$. However, such an approximation would be very poor for designs with $\eta^\ast_l \approx \eta^\ast_j$ for some $l\neq j$. Our novel asymptotic framework allows consideration of the case where there is near-replication of linear predictor values in a block.

Assumptions~1 and~2 allow the partition of $\mathcal{S}=\{ 1,\ldots,m \}$ as $\mathcal{N}(j)\cup \mathcal{Z}(j) \cup \mathcal{P}(j)$ for each $j=1,\ldots,m$, where $\mathcal{N}(j) = \{ l : \eta_l-\eta_j \to -\infty \}$, $\mathcal{Z}(j) = \{ l : \eta_l - \eta_j \to 0 \}$, and $\mathcal{P}(j)=\{ l : \eta_l - \eta_j \to \infty\}$. Intuitively, $\mathcal{N}(j)$, $\mathcal{Z}(j)$, $\mathcal{P}(j)$ are the respective sets of indices of linear predictors less than, similar to, and greater than $\eta_j$. The limiting expressions we develop for $P(Y|\theta,\zeta)$ and $\partial P(Y|\theta,\zeta)/\partial \eta_j$ depend on which elements of $\mathcal{S}$ belong to $\mathcal{N}$, $\mathcal{Z}$ and $\mathcal{P}$. 

It will be useful to identify some particular classes of outcomes.
\begin{definition}
Outcome $Y=(y_1,\ldots,y_m)^\T$ is \emph{increasing} (within the block) if there exists $j'\in\mathcal{S}$ such that $y_l=0$ when $\eta_l-\eta_{j'}<0$ and $y_l=1$ when $\eta_l-\eta_{j'}>0$.
\end{definition}
\begin{definition}
Outcome $Y$ is \emph{quasi-increasing} if there exists $j'\in\mathcal{S}$ such that $\mathcal{N}(j') \subseteq \mathcal{S}_0$ and $\mathcal{P}(j')\subseteq \mathcal{S}_1$,  where $\mathcal{S}_0 = \{ j : y_j=0 \}$ and $\mathcal{S}_1 = \{ j : y_j=1\}$,
 or, equivalently, if $\{ \mathcal{S}_0 \cap \mathcal{P}(j')\} \cup \{ \mathcal{S}_1 \cap \mathcal{N}(j')\}=\emptyset$.
\end{definition}
Any outcome that is increasing (with the same $j'$) for all $\sigma^2$ is clearly also quasi-increasing.   

We now make a further assumption necessary for our theorems.%the following asymptotic approximations for the likelihood and its derivative. 

%It is useful to identify some particular classes of outcomes. Let $Y=(y_1,\ldots,y_m)^\T$ be an arbitrary outcome, and  define $\mathcal{S}_0 = \{ j : y_j=0 \}$ and $\mathcal{S}_1 = \{ j : y_j=1\}$. We say that $Y$ is \emph{increasing} if there exists $\tilde{\eta}$ such that $y_j=0$ when $\eta_j< \tilde{\eta}$, and $y_j=1$ when $\eta_j > \tilde{\eta}$. We say that $Y$ is \emph{quasi-increasing} if  there exists some $j'$ such that $\mathcal{N}(j') \subseteq \mathcal{S}_0$ and $\mathcal{P}(j')\subseteq \mathcal{S}_1$ or, equivalently, if $\{ \mathcal{S}_0 \cap \mathcal{P}(j')\} \cup \{ \mathcal{S}_1 \cap \mathcal{N}(j')\}=\emptyset$. 

%characterize the asymptotic behaviour of model probabilities.

\begin{assumption}
There exists $A_j, B_j>0$ such that $|\eta_l -\eta_j| > \sigma A_j$ for $l \in \{\mathcal{S}_0\cap \mathcal{N}(j)\} \cup \{\mathcal{S}_1 \cap \mathcal{P}(j)\}$ and $|\eta_l - \eta_j| > \sigma B_j$ for all $l \in \{ \mathcal{S}_1 \cap \mathcal{N}(j) \}\cup\{  \mathcal{S}_0 \cap \mathcal{P}(j)  \}$. 
\end{assumption}
This condition holds for large $\sigma^2$ by Assumptions 1 and 2; it implies that pairs of predictors which diverge asymptotically are at least $\min_{j=1,\ldots,m}\{A_j,B_j\}\sigma$ apart. 

\begin{theorem}[Approximation of the likelihood]
\label{thm:asymp-prob-quasi-incr}
Suppose that the outcome is quasi-increasing. Then there exists $j'\in\mathcal{S}$ such that $\{\mathcal{S}_0 \cap \mathcal{P}(j')\} \cup \{\mathcal{S}_1 \cap \mathcal{N}(j')\} = \emptyset$, and:\\
(i) If $|\mathcal{S}_0 \cap \mathcal{Z}(j')|=0$ or $|\mathcal{S}_1 \cap \mathcal{Z}(j')|=0$, the outcome is increasing and, as $\sigma^2\to \infty$,
 \begin{equation}\label{increasinglik}
P(Y|\theta,\zeta) = \max\left\{0,
		\Phi\left(-\max_{j \in \mathcal{S}_0} \{ \eta_j/\sigma \}\right) - \Phi\left(-\min_{j \in \mathcal{S}_1}\{\eta_j/\sigma\}\right) 
		\right\} + O(\sigma^{-1})  \,.
\end{equation}
\\
(ii) If  $|\mathcal{S}_0 \cap \mathcal{Z}(j')|\geq 1 $ and $|\mathcal{S}_1 \cap \mathcal{Z}(j')|\geq 1$, then as $\sigma^2 \to \infty$,
\begin{align}
P(Y|\theta, \zeta) &= \frac{\phi(\eta_{j'}/\sigma)}{\sigma}
			\int_{-\infty}^{\infty} 
			\{1-h(t)\}^{|\mathcal{S}_0 \cap \mathcal{Z}(j')|} 
			h(t)^{|\mathcal{S}_1 \cap \mathcal{Z}(j')|} 
			dt \nonumber\\ 
			& \quad\quad\quad\quad\quad\quad\quad+ \sum_{l\in\mathcal{Z}(j')}O(\Delta_{lj'}/\sigma) + O(\sigma^{-2}) \,,\label{qincreasinglik}
\end{align}
where $\Delta_{lj} = \eta_l - \eta_j$. The integral has value 1 when $|\mathcal{Z}(j')|=2$.
%If $|\mathcal{Z}(j')|=2$, then $|\mathcal{S}_0 \cap \mathcal{Z}(j')| = |\mathcal{S}_1 \cap \mathcal{Z}(j')|=1$ and the integral above simplifies to 1.
\end{theorem}

\begin{theorem}[Approximation of derivatives]
\label{thm:asymp-approx-deriv}
(i) Assume $j\in\mathcal{S}$ is such that  $\mathcal{N}(j) \subseteq \mathcal{S}_0$ and $\mathcal{P}(j)\subseteq \mathcal{S}_1$ which implies that the outcome is quasi-increasing. Then, for arbitrary $\epsilon >0$,
\begin{align}
(2y_j-1) \frac{\partial P(Y|\theta,\zeta)}
{\partial \eta_j} &=
					 \frac{1}{\sigma}\phi\left(\frac{-\eta_j}{\sigma}\right)
					 \Big\{
					  C^{(1)}_{I(j),J(j)}
					%\\
					 %& \quad\quad\quad\quad\quad\quad\quad
					 + \sum_{l\in \mathcal{S}_1 \cap \mathcal{Z}(j)\backslash\{ j\}}
				\Delta_{lj}  C^{(3)}_{I(j)-1,J(j)}  \nonumber\\
				& \quad\quad\quad
				 	-  \sum_{l\in \mathcal{S}_0 \cap \mathcal{Z}(j)\backslash\{ j\}}
				\Delta_{lj} C^{(3)}_{I(j),J(j)-1} \Big\}
				+  \frac{1}{\sigma^2}\phi'\left(\frac{-\eta_j}{\sigma}\right) C^{(2)}_{I(j),J(j)} \nonumber
				\\[0.5ex]&\quad\quad\quad
				%\\[1ex]
					 %
					 + \sum_{l \in \mathcal{Z}(j)}O(\Delta_{lj}^2/\sigma) + O(\sigma^{-3}) +  O\left(\frac{1}{\sigma}e^{-\sigma A_j/[(1+\epsilon)m]}\right) \,, \label{qincreasingderiv}
\end{align}
where $I(j)=|\mathcal{S}_1\cap\mathcal{Z}(j)\backslash\{j\}|$, $J(j)=|\mathcal{S}_0\cap\mathcal{Z}(j)\backslash\{j\}|$ and, for integers $I,J\geq 0$,
\begin{alignat*}{2}
C^{(1)}_{I,J} =  \int_{-\infty}^{\infty} 
h'(t)  h(t)^{I}\{1-h(t)\}^{J} dt\,,  & \hspace{0.5cm}
C^{(2)}_{I,J}= \int_{-\infty}^{\infty}  t h'(t)  h(t)^{I}\{1-h(t)\}^{J} dt \\
C^{(3)}_{I,J}= \int_{-\infty}^{\infty} \{h'(t)\}^2  h(t)^{I}\{1-h(t)\}^{J} dt \,.
\end{alignat*}

(ii) If $j$ is such that $\mathcal{N}(j) \not\subseteq \mathcal{S}_0 $ or $\mathcal{P}(j) \not\subseteq \mathcal{S}_1 $, then the derivative satisfies
\begin{equation}
\frac{\partial P(Y|\theta,\zeta)}
{\partial \eta_j} = O\left(\frac{1}{\sigma}e^{-\sigma B_j/(1+\epsilon)}\right)  \,.
\label{eq:appx-deriv-fast-cvg}
\end{equation}
Moreover, if the outcome is neither increasing nor quasi-increasing, then~\eqref{eq:appx-deriv-fast-cvg} holds for all $j$.\end{theorem}

\begin{theorem}[Importance of quasi-increasing outcomes]\label{thm:quasi-only}
Quasi-increasing outcomes contribute terms of order $O(\sigma^{-1})$ or $O(\sigma^{-2})$ to the information matrix in~\eqref{eq:blockinf}. Outcomes that are not quasi-increasing contribute terms of order $O\left(\frac{1}{\sigma}e^{-\sigma \min_j B_j/(1+\epsilon)}\right)$ which are asymptotically negligible. 
%For all outcomes,
%\[
%\left|\frac{\partial P(Y|\theta,\zeta)}{\partial \eta_j}\right| / P(Y|\theta,\zeta) \leq 2 \,.
%\]
%As a result, the contribution to the information matrix from an outcome that is neither increasing nor quasi-increasing is $O\left(\frac{1}{\sigma}e^{-\sigma \min_j B_j/(1+\epsilon)}\right)$. Quasi-increasing outcomes contribute terms with entries of order $O(\sigma^{-1})$ and $O(\sigma^{-2})$.
\end{theorem}

Proofs for Theorems 1--3 are in Appendix~2. From Theorem~\ref{thm:quasi-only}, an asymptotic approximation to the information matrix in \eqref{eq:blockinf} need only include contributions from increasing and quasi-increasing outcomes. The importance of quasi-increasing outcomes seems difficult to capture using approximations, such as those in Sections~\ref{sec:mql}--\ref{sec:amql}, that only incorporate the first and second order moments of the joint distribution of the responses.
 
We combine the asymptotic approximations to $P(Y|\theta,\zeta)$ and $\partial P(Y|\theta,\zeta)/\partial \eta_j$, from Theorems 1 and 2 respectively, with~\eqref{eq:blockinf} to provide an \textit{asymptotic outcome-enumeration} approximation to the information matrix. The only additional requirement is that for the $j$th treatment in the $i$th support block $(i=1,\ldots,b;\,j=1,\ldots,m)$, we must define a suitable partition of the indices $\{1,\ldots,m\}$ into sets $\mathcal{N}_i(j)$, $\mathcal{Z}_i(j)$, and $\mathcal{P}_i(j)$. 
%This partition should be  such that: linear predictors $\eta_{ij}$, $\eta_{ij'}$ with $j,j' \in \mathcal{Z}_{i}(j)$ are `close' (relative to $\sigma$); $\eta_{ij}$, $\eta_{ij'}$ with $j' \in \mathcal{N}_i(j)$ are such that $\eta_{ij'}$ is `significantly' (relative to $\sigma$) smaller than $\eta_{ij}$; and $\eta_{ij}$, $\eta_{ij'}$ with $j' \in \mathcal{P}_i(j)$ are such that $\eta_{ij'}$ is `significantly' (relative to $\sigma$) larger than $\eta_{ij}$. 
This partition should be such that linear predictor $\eta_{ij'}$, relative to $\sigma$, is close to, less than, or greater than $\eta_{ij}$ for $j' \in \mathcal{Z}_{i}(j)$, $j' \in \mathcal{N}_i(j)$ or $j' \in \mathcal{P}_i(j)$ respectively. We propose to form these partitions automatically using the heuristic algorithm given in Appendix~1. 
%A key problem when forming the approximation is to decide automatically which predictors should be treated as `close' i.e.~given $\zeta$ and $j$, which indices should we assume belong to $\mathcal{N}(j)$, $\mathcal{Z}(j)$ and $\mathcal{P}(j)$? A heuristic algorithm to answer this question is given in the Appendix. 
%Computationally, our strong dependence approximation involves standard normal density and distribution calculations, which have mature implementations; general Gauss-Hermite quadrature is not needed. This contrasts with outcome enumeration, where the number of quadrature points must be increased as $\sigma^2$ increases in order to maintain accuracy as the integrands converge to step functions. 
As presently implemented, the objective function corresponding to the asymptotic approximation is discontinuous; nonetheless, the resulting designs typically have high efficiencies relative to designs from the na\"ive outcome-enumeration approximation, competitive with those from the other methods. In certain circumstances, discussed in Section~\ref{sec:ex-local}, this asymptotic approximation outperforms other methods.

When $\sigma^2$ is large, the recovery of inter-block information that occurs when using a mixed model for analysis is important for parameter estimation. For large $\sigma^2$, separation of outcomes \citep{aa1984} occurs within all blocks with high probability, in which case the parameters of the corresponding fixed block effects model are not estimable (see Propositions 1 and 2 in Appendix~2). Despite this, efficient parameter estimation is still possible under the mixed model (see Appendix~3).

\subsection{Interpolated outcome-enumeration}
\label{sec:mlni}

In this section, we discuss a more direct numerical approximation for $M_\beta(\xi,\theta)$ under the logistic random intercept model. Note that for this model, the likelihood depends on the regression parameters $\beta$ only through the vector $\eta$. Let $P_Y(\eta,\sigma^2)= P(Y|\zeta, \theta)$ and define 
\begin{equation}
\mathcal{Q}(\eta, \sigma^2)= \sum_{Y \in \{0,1\}^m} 
					\frac{1}{P_Y} 
						\left( 
							\frac{\partial P_Y}
								{\partial \eta}
						\right) 
						\left( 
							\frac{\partial P_Y}
								{\partial \eta}
						\right)^\T \,,
						\label{eq:weight}
\end{equation}
so that, by \eqref{eq:blockinf}, $M_\beta(\zeta,\theta)= F^\T \mathcal{Q} F$.

An \textit{interpolated outcome-enumeration} approximation to $M_\beta(\zeta,\theta)$ can be developed by surrogate modelling of the matrix-valued function $\mathcal{Q}$. The idea is to compute the values of the function $\mathcal{Q}$ at a collection of training points, and interpolate these data to predict the value of $\mathcal{Q}$ at new sites $(\eta,\sigma^2)$. Interpolating $\mathcal{Q}$ as a function of $\eta$ is particularly computationally efficient for finding Bayesian designs, as the same interpolator can be used for any value of $\beta$.

Surrogate modelling is widely applied in `computer experiments' on expensive-to-evaluate computational models for complex phenomena (see \citealp{santner}). We believe that its use for accelerating the computation of approximations necessary for the optimal design of physical experiments is new.  For computer experiments, Gaussian process modelling (Kriging) is well-established as a surrogate; it can be used with training sets not arranged in a regular grid and can straightforwardly be applied to multidimensional problems. For block size $m=2$, it is faster to use bilinear or bicubic interpolation and a regular grid.

\section{Examples for binary response}
\label{sec:bin}
\subsection{Preliminaries for the examples}
In Sections \ref{sec:ex-local} and \ref{sec:EX1}, $D$-optimal designs are found, compared and assessed for blocks of size $m=4$ and a binary response logistic random intercept model with two variables and the following linear predictor
\begin{equation}
\label{eq:2factor-model}
\nu(x;u,\beta) = \beta_0 + \beta_1 x^{(1)} + \beta_2 x^{(2)} + u \,, \quad u \sim N(0,\sigma^2)\,,
\end{equation}
%\noindent with blocks of size $m=4$.
where $x=(x^{(1)}, x^{(2)})^\T\in [-1,1]^2$. In Section \ref{sec:4factor}, $D$-optimal designs are found for a logistic random intercept model with four factors and eight fixed parameters.

In Sections~\ref{sec:ex-local} and~\ref{sec:4factor}, we find locally $D$-optimal designs for various parameter scenarios by approximating the information matrix using adjusted generalized estimating equations and adjusted marginal quasi-likelihood. In Section~\ref{sec:ex-local} we also find locally optimal designs using unadjusted generalized estimating equations (assuming $\beta^\ast=\beta$) and, for large $\sigma^2$, asymptotic outcome-enumeration. In these sections we find it advantageous to specify parameter scenarios on the scale of the marginal effects, $\beta_\text{att}$, to facilitate performance comparisons across different values of $\sigma^2$. Intuitively, this setup mimics strong information being available for the marginal effects, and uncertainty in the strength of dependence. In Section~\ref{sec:EX1}, we find Bayesian $D$-optimal designs, with the prior information specified on the conditional parameters, as no comparisons are made across different values of $\sigma^2$; we set $\beta_0 \sim U[-0.5,0.5]$, $\beta_1 \sim U[3,5]$, $\beta_2 \sim U[0,10]$, and $\sigma^2=5$. Thus, there is substantial uncertainty in the value of $\beta_2$, and moderate block-to-block variability. Here, we approximate the information matrix using the adjusted marginal quasi-likelihood, adjusted generalized estimating equations, and interpolated outcome-enumeration methods. For all of our examples, efficiencies of optimal designs found using the different approximations are calculated relative to $D$-optimal designs found using the na\"ive outcome-enumeration approximation.

For all approximations, we use a quasi-Newton method \citep[the Broyden--Fletcher--Goldfarb--Shanno algorithm;][pp.~136--143]{nocedal-wright} to obtain optimal, or near-optimal, designs numerically; that is optimal or highly efficient combinations of $\zeta_k$, $w_k$ and $b$. Multiple random starts of the algorithm are used to attempt to identify a global optimum of the objective function. Convergence is assessed via comparison of the optima obtained from the different starts, and was considered satisfactory for the examples presented here. We assess performance of the obtained designs using local efficiency, $
\mbox{eff}(\xi | \theta) = \{
 	{|M_\beta(\xi,\theta)|} /
 		{\sup_{\xi'} |M_\beta(\xi' ,\theta)|}
\}^{1/p} $.

%In Sections~\ref{sec:ex-local} and~\ref{sec:4factor}, we find it advantageous to specify parameter scenarios on the scale of the marginal effects, $\beta_\text{att}$. This facilitates performance comparisons across different values of $\sigma^2$. Intuitively, this setup mimics strong information being available for the marginal effects, and uncertainty in the strength of dependence. In Section~\ref{sec:EX1}, prior information is specified about the conditional parameters directly, as no comparisons are made across different values of $\sigma^2$. 

\subsection{Example 1: Locally optimal designs}
\label{sec:ex-local}

%Let $x=(x^{(1)}, x^{(2)})^\T$. In this section we calculate locally optimal designs for the logistic random intercept model with linear predictor
%\begin{equation}\label{eq:2factor-model}
%\nu(x;u,\beta) = \beta_0 + \beta_1 x^{(1)} + \beta_2 x^{(2)} + u \,, \quad u_i \sim N(0,\sigma^2)\,,
%\end{equation}
%for various values of the parameter vector $\theta= (\beta_0,\beta_1,\beta_2,\sigma^2)^\T$. 

\begin{table}[p]
{\footnotesize
\begin{tabular}{ccccccccc}
 && \multicolumn{6}{c}{$\sigma^2$} \\
$\beta_\text{att}^{\T}$ & Design & 1& 2& 5& 10& 20 & 50 \\[0.7ex]
(0,1,1) &  
Unadj.~gen. & 96$\cdot$3--100$\cdot$0  & 94$\cdot$5--100$\cdot$1 & 82$\cdot$9--99$\cdot$6 &  78$\cdot$8--94$\cdot$3 & 71$\cdot$8--87$\cdot$1 &  59$\cdot$9--76$\cdot$6	\\[0.5ex]
& Adj.~marg. & 100 & 100 & 100 & 100 & 100 & 100\\
 &Adj.~gen. &  99$\cdot$7--100$\cdot$0 &  99$\cdot$7--100$\cdot$1 & 99$\cdot$2--100$\cdot$0 & 98$\cdot$8--100$\cdot$0 & 98$\cdot$5--100$\cdot$0 &98$\cdot$1--100$\cdot$0 \\
  &Asymp.~enum.& &      &    && 100$\cdot$0&  94$\cdot$8\\[0.5ex]
(0,3,2)&   
Unadj.~gen. &  86$\cdot$2--97$\cdot$3 & 84$\cdot$5--93$\cdot$2  & 79$\cdot$1--85$\cdot$3 &  74$\cdot$7--79$\cdot$0 &70$\cdot$9--73$\cdot$4 &  63$\cdot$3--67$\cdot$7  \\[0.5ex]
&Adj.~marg.&  99$\cdot$9 &99$\cdot$9 &100$\cdot$0 & 99$\cdot$9 & 99$\cdot$4  &95$\cdot$2 \\
& Adj.~gen.& 85$\cdot$3--99$\cdot$8 & 85$\cdot$6--99$\cdot$6 & 86$\cdot$3--99$\cdot$5 & 87$\cdot$2--99$\cdot$7 & 87$\cdot$7--100$\cdot$0 &  83$\cdot$9--98$\cdot$5\\
  & Asymp.~enum.  &  &  &  &   &96$\cdot$4& 97$\cdot$4\\[0.5ex]
 (0,5,10)& 
 Unadj.~gen. & 82$\cdot$3--96$\cdot$1 & 79$\cdot$8--91$\cdot$9  & 70$\cdot$1--84$\cdot$4 & 65$\cdot$2--78$\cdot$4 & 
64$\cdot$7--72$\cdot$8 & 52$\cdot$2--67$\cdot$0\\[0.5ex]
  \quad $\times (1+5c^2 )^{-\frac{1}{2}}$ & 
  Adj.~marg. & 99$\cdot$9  & 99$\cdot$9 & 100$\cdot$0 & 99$\cdot$8  & 99$\cdot$7&   99$\cdot$8\\ 
 &Adj.~gen.& 83$\cdot$9--99$\cdot$1 & 84$\cdot$1--98$\cdot$7 & 84$\cdot$8--98$\cdot$9 & 85$\cdot$6--99$\cdot$4 & 86$\cdot$0--99$\cdot$5 & 83$\cdot$9--99$\cdot$1\\
& Asymp.~enum. &    &  &    &  & 94$\cdot$8&  96$\cdot$1\\[0.5ex]
(1,2,3)$\dagger$     &
Unadj.~gen. & 84$\cdot$3--96$\cdot$8 & 83$\cdot$2--94$\cdot$5 & 75$\cdot$2--88$\cdot$5 & 70$\cdot$6--78$\cdot$7 & 65$\cdot$7--78$\cdot$9 & 58$\cdot$9--73$\cdot$9 \\[0.5ex]
& Adj.~marg. 			&100$\cdot$4 	& 99$\cdot$1  	& 96$\cdot$6 	& 92$\cdot$1 	& 84$\cdot$8 	& 73$\cdot$5(*)\\
	 & Adj.~gen. 			& 84$\cdot$1--99$\cdot$5 & 85$\cdot$0--99$\cdot$0 & 86$\cdot$6--99$\cdot$2 & 87$\cdot$6--98$\cdot$9 & 86$\cdot$7--97$\cdot$2 & 77$\cdot$6--91$\cdot$7(*) \\
 &Asymp.~enum. &    & &    &&  93$\cdot$5&  98$\cdot$3 \\[0.5ex]
(1,4,4) &
Unadj.~gen. & 81$\cdot$7--96$\cdot$9 &  80$\cdot$7--94$\cdot$1& 76$\cdot$0--86$\cdot$3 &70$\cdot$8--79$\cdot$8& 65$\cdot$0--73$\cdot$1& 58$\cdot$3--64$\cdot$6\\[0.5ex]
& Adj.~marg. & 100$\cdot$0 & 100$\cdot$0 & 99$\cdot$9  & 99$\cdot$4  & 98$\cdot$2  & 97$\cdot$2\\
& Adj.~gen. &  80$\cdot$3--99$\cdot$4 & 81$\cdot$0--99$\cdot$1 & 81$\cdot$9--99$\cdot$3 & 82$\cdot$5--99$\cdot$7 & 82$\cdot$5--98$\cdot$9 & 82$\cdot$5--97$\cdot$4 \\
& Asymp.~enum. &    &  & &   & 97$\cdot$1&  98$\cdot$2\\[0.5ex]
(1,3,3) &
Unadj.~gen.  &82$\cdot$2--97$\cdot$1& 80$\cdot$6--93$\cdot$5  &76$\cdot$4--86$\cdot$4 &71$\cdot$0--79$\cdot$7& 65$\cdot$0--72$\cdot$7& 57$\cdot$1--63$\cdot$1\\[0.5ex]
&Adj.~marg. & 99$\cdot$7 & 99$\cdot$5 & 100$\cdot$0  & 99$\cdot$3 & 97$\cdot$9  & 95$\cdot$1\\
&Adj.~gen.&  79$\cdot$6--99$\cdot$3 & 80$\cdot$5--98$\cdot$6 & 83$\cdot$1--99$\cdot$4 & 84$\cdot$9--99$\cdot$7 & 85$\cdot$9--98$\cdot$8 & 84$\cdot$6--95$\cdot$6 \\
&Asymp.~enum. &  & &    &   &97$\cdot$1&  98$\cdot$2 \\[0.5ex]
(1,2,2) &    
 Unadj.~gen.  & 82$\cdot$1--97$\cdot$9 & 81$\cdot$9--93$\cdot$4 & 78$\cdot$6--88$\cdot$2 & 73$\cdot$4--78$\cdot$6&  67$\cdot$3--69$\cdot$5 & 58$\cdot$4--64$\cdot$3 \\[0.5ex]
&Adj.~marg. & 100$\cdot$6 & 100$\cdot$3 & 100$\cdot$2 & 98$\cdot$5  & 96$\cdot$1  & 92$\cdot$6 \\
& Adj.~gen. & 83$\cdot$3--100$\cdot$1 & 84$\cdot$6-- 99$\cdot$3 &  87$\cdot$4-- 99$\cdot$8 & 89$\cdot$3-- 99$\cdot$6 & 90$\cdot$1-- 98$\cdot$0 &  88$\cdot$5-- 95$\cdot$4 \\
& Asymp.~enum. &     & &    &   & 95$\cdot$1& 97$\cdot$9 \\[0.5ex]
(2,1,3)$\dagger$ &   
Unadj.~gen. & 84$\cdot$3--96$\cdot$8 & 83$\cdot$2--94$\cdot$5&78$\cdot$5--89$\cdot$2 & 73$\cdot$4--85$\cdot$9 & 62$\cdot$8--76$\cdot$7 & 56$\cdot$4--64$\cdot$1	\\[0.5ex]
 &Adj.~marg. & 99$\cdot$9 & 99$\cdot$1 & 96$\cdot$6& 92$\cdot$1& 84$\cdot$8& 78$\cdot$0(*)\\
& Adj.~gen.&  83$\cdot$6-99$\cdot$5 & 84$\cdot$5-99$\cdot$0 & 86$\cdot$1-99$\cdot$1 & 86$\cdot$8-98$\cdot$9 & 85$\cdot$8-97$\cdot$0 & 77$\cdot$4-90$\cdot$9(*)\\
& Asymp.~enum.& &  & &  & 95$\cdot$0&  98$\cdot$0
\end{tabular}}
\caption{Example 1: computed efficiencies of locally $D$-optimal designs from different methods\label{tab:ex-local}. Unadj.~gen. -- Unadjusted generalized estimating equations; Adj.~marg. -- adjusted marginal quasi-likelihood; Adj.~gen. -- adjusted generalized estimating equations;  Asymp.~enum. -- asymptotic outcome-enumeration. Reported efficiencies for generalized estimating equation methods are for $\rho$ = 0$\cdot$1, 0$\cdot$15, 0$\cdot$2, $\ldots$, 0$\cdot$7. Symbols $\dagger$ and (*) indicate parameter values for which the adjusted marginal modelling approximations give particularly inefficient designs.}
\end{table}

\begin{table}[p]
\begin{tabular}{cc}
Method &  Time per parameter vector (processor-seconds)\\
Na\"ive outcome-enumeration ($\sigma^2=50$)  & $3\times 10^5$  \\
Na\"ive outcome-enumeration ($\sigma^2=1$)& $5\times 10^4$ \\ 
Asymptotic outcome-enumeration &   $7\times 10^3$  \\
Adjusted generalized estimating equations*& $1 \times 10^3$ ($8\times 10^3$)  \\
 Adjusted marginal quasi-likelihood & $4\times 10^2$\\
\end{tabular}
\caption{Example 1: computational expense for locally $D$-optimal designs\label{tab:local-comptime}. *Figure in brackets is indicative of time when design re-use is impossible.}
\end{table}

The purpose of this example is twofold. Firstly, we wish to illustrate the performance of the methods for different $\sigma^2$. Secondly, we demonstrate circumstances under which the resulting designs are robust to a reasonable range of values assumed for $\sigma^2$.

Table \ref{tab:ex-local} gives the efficiencies under this regime of optimal designs from the different approximations relative to an optimal design found using the na\"ive outcome-enumeration approximation.  It is clear that the unadjusted generalized estimating equation approach is by far the worst method, with efficiencies frequently less than 90\%. In most cases, the remaining closed-form approximations are competitive with na\"ive outcome-enumeration. The performance of the adjusted generalized estimating equation approach depends critically on the choice of $\rho$ which is treated here as a tuning parameter.

We observed two cases for which the adjusted marginal and adjusted generalized estimating equation methods performed poorly. For $\beta_\text{att} = (1,2,3)^\T$ and $(2,1,3)^\T$, with $\sigma^2=50$, the design efficiencies from the former two methods were below 92\%. These cases are unusual in that, for all $\sigma^2>1$, the two marginal approximations selected designs that replicate treatments within at least one of their blocks. This appears inefficient: the designs from both the na\"ive and asymptotic outcome-enumeration approximations do not feature within-block replication, and the latter design is at least 98\% efficient. The only other case where this replication occurred in the marginal approximation designs for large $\sigma^2$ was $\beta_\text{att}=(1,2,2)^\T$, where the efficiency was again relatively low. Our theoretical results (Section~\ref{sec:strongdep}) suggest that marginal methods may poorly approximate the information matrix for designs featuring within-block replication when $\sigma^2$ is large. Thus we would recommend some caution when $\sigma^2$ is large and use of the marginal approximations yields designs featuring within-block replication of treatments. For such designs, the error from these approximations may be large. Additionally, the small $u$ Taylor approximations underlying the covariance approximation in the adjusted marginal quasi-likelihood method cannot be expected to be accurate when $\sigma^2$ is large and large random effects are anticipated.
%\textbf{Comments on the table: all methods are competitive in many cases. Situations where AMQL and GEE perform poorly - seem to occur when there is replication within blocks? In this case, the asymptotic approximation should be better because of the work we have done to specifically to account for (near-)replication?}

Table \ref{tab:local-comptime} gives the average total processor time for each method, as recorded in a high performance parallel computing environment with twelve 2$\cdot$4GHz cores per node. The times given are per parameter vector for 100 random starts of the optimization algorithm. Na\"ive outcome-enumeration is the most expensive method followed by asymptotic outcome-enumeration, adjusted generalized estimating equations and adjusted marginal quasi-likelihood. The computational expense of the adjusted generalized estimating equations method depends on the structure of the problem. Here, there are many parameter scenarios with the same values of $\beta_\text{att}$ which allows re-use  of adjusted generalized estimating equation designs for a given $\beta_{\text{att}}$ for various $\sigma^2$. If re-use were not possible, then the time per design would be higher: an indicative figure is given in parentheses. The time to obtain a design for given $\rho$ is comparable with that from adjusted marginal quasi-likelihood.

For moderate dependence ($\sigma^2 \leq 10$) choosing a single value of $\sigma^2=5$ appears to be very robust; for all $\beta_\text{att}$ considered, the na\"ive outcome-enumeration design with $\sigma^2 =5$ has a calculated efficiency of at least 99$\cdot$1\% for $\sigma^2=1,2,10$. Assuming a single value of $\sigma^2=1$ is less robust, though still reasonable, the worst case is when $\beta_\text{att}=(1,2,2)^\T$ and the true $\sigma^2=10$; the efficiency of the resulting design is 97$\cdot$1\%. 
However, if the dependence is actually strong then the above designs may perform comparatively poorly; when $\beta_\text{att}=(0,3,2)^\T$, $\sigma^2=50$, the design obtained assuming $\sigma^2=5$ has a calculated $D$-efficiency of 93$\cdot$6\%. This robustness of an optimal design to a wide range of assumed values of $\sigma^2$ is a consequence of specifying the parameters on the marginal scale.

\subsection{Example 2: Bayesian optimal designs}
\label{sec:EX1}

\begin{figure}[t]
\begin{center}
\includegraphics[width=0.8\linewidth]{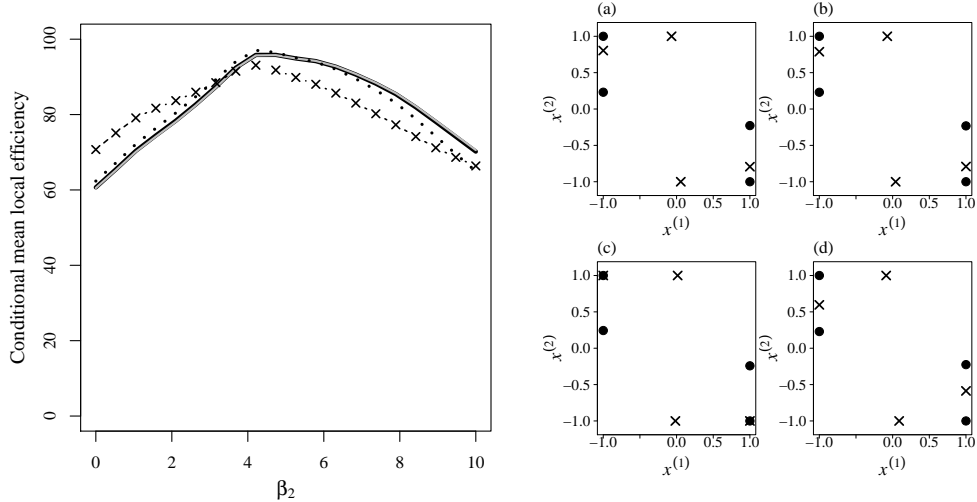}
\caption{Example 2: conditional mean efficiencies $\E[\eff(\xi;\theta)|\beta_2]$ (left panel) and support blocks (right panel) of Bayesian $D$-optimal designs from maximum likelihood via na\"ive outcome-enumeration  [black line, (a)],  maximum likelihood via interpolated outcome-enumeration [solid grey line, (b)], adjusted marginal quasi-likelihood [dotted black line, (c)], and adjusted generalized estimating equations [dashed and crossed black line, (d)]. Treatments with the same plotting character are in the same block.}
\end{center}
\label{fig:lookup-bayes-effs}
\end{figure}

%Let $x=(x^{(1)}, x^{(2)})^\T$. In this section we calculate Bayesian optimal designs for the logistic random intercept model with linear predictor
%\begin{equation}\label{eq:2factor-model}
%\nu(x;u,\beta) = \beta_0 + \beta_1 x^{(1)} + \beta_2 x^{(2)} + u \,, \quad u_i \sim N(0,\sigma^2)\,,
%\end{equation}
%\noindent where a priori $(\beta_0,\beta_1)=(0,5)$, $\sigma^2=5$ and $\beta_2 \sim U[0,10]$. Thus, we have substantial uncertainty as to the value of $\beta_2$, and reasonably substantial block-to-block variability.

%\begin{figure}[tbhp]
%\begin{centering}
%\begin{minipage}[b]{0.49\linewidth}
%\figurebox{3in}{}{}[ex2-effs.eps]
%%\includegraphics[scale=0.5]{art/ex2-effs.eps}
%\end{minipage}
%\begin{minipage}[b]{0.49\linewidth}
%\figurebox{3in}{}{}[ex1-Bayesian-design-points.eps]
%%\includegraphics[scale=0.5]{art/ex1-Bayesian-design-points.eps}
%\end{minipage}
%\end{centering}
%\caption{Local efficiencies (left panel) and support points (right panel) of Bayesian designs from computational approximation [black line, (a)], adjusted marginal quasi-likelihood [dashed black line, (b)], marginal quasi-likelihood [dashed grey line, (c)] and of the factorial design [dotted black line, (d)]. Pairs of treatments with the same plotting character are in the same block.}
%\label{fig:lookup-bayes-effs}
%\end{figure}

We choose the design $\xi$ to maximize $\psi(\xi)$ from Section~\ref{sec:opt-crit} and approximate the integral in the objective function by averaging over a Latin hypercube sample of 50 values of $\beta$ from $[-0.5,0.5]\times[3,5]\times[0,10]$. As the value of $\sigma^2$ is assumed known, for the interpolated outcome-enumeration approximation we only need build a surrogate model of $\mathcal{Q}$ as a function of $\eta$.

Bayesian $D$-optimal designs were computed for each of the different approximations using 1000 random starts; the support blocks of the designs are shown in Fig.~\ref{fig:lookup-bayes-effs} with corresponding weights given in Table~\ref{tab:lookup-example-designs}. A single value $\rho=0$$\cdot$$6$, corresponding to fairly strong correlation, was used in the adjusted generalized estimating equations approximation. For each method, from the 1000 designs generated the best was selected with respect to na\"ive outcome-enumeration. 

All of the designs contain multiple support blocks due to the degree of uncertainty in the parameters and the small block size. Locally $D$-optimal designs were also found for each of the 50 sampled parameter vectors under the na\"ive outcome-enumeration approximation, and the local efficiency of each Bayesian design was calculated relative to each of these 50 designs. Then, Gaussian process emulators were constructed for the efficiency profile of each Bayesian design. Figure \ref{fig:lookup-bayes-effs} shows the dependence of the efficiency on $\beta_2$, via approximations of $\E[ \eff(\xi;\theta) | \beta_2]$ obtained from the efficiency profile emulators. The performance of all of the Bayesian designs varied little according to the value of $\beta_0$ or $\beta_1$, with $\E[\eff(\xi;\theta)|\beta_0]$ and $\E[\eff(\xi;\theta)|\beta_1]$ changing by fewer than 4 percentage points over the ranges of $\beta_0$ and $\beta_1$ respectively. The conditional mean efficiency of the design from the adjusted generalized estimated equations approach is clearly quite different, as a function of $\beta_2$, from the local efficiencies from the other methods. The designs from all of the approximations appear similar to the na\"ive outcome-enumeration design (compare Figures~\ref{fig:lookup-bayes-effs}(a)--(d)).

To train the interpolated outcome-enumeration approximation of $\mathcal{Q}$, a random Latin hypercube sample of 10,000 $\eta$ vectors was drawn from $[-20,20]^4$, and the matrix $\mathcal{Q}$ evaluated for each vector. The second-order, compactly-supported Wendland covariance function was used, with range parameter chosen manually as 15 to make the predictions appear reasonably smooth and accurate. Independent Gaussian process models were fitted to the $m^2$ entries of $\mathcal{Q}$. The use of a compactly-supported covariance function is advantageous here due to the large number of training points; it enables inversion of the covariance matrix in a reasonable time, and permits relatively fast predictions from the fitted model. For finding Bayesian designs, the interpolation method required around 3$\cdot$2 times less computational effort than na\"ive outcome-enumeration for this example (Table \ref{tab:lookup-example-designs}). If more quadrature points were used to approximate the prior distribution, or if an adequate emulator could be found using fewer training points, then the advantage of using interpolation to approximate the objective function would be greater (for 200 quadrature points, with the same training set, objective function evaluation using interpolation is approximately 6 times faster than na\"ive outcome-enumeration). The advantage will also be more pronounced for larger $\sigma^2$. The closed-form approximations (using a single $\rho$) are approximately two orders of magnitude faster than na\"ive outcome-enumeration.

\begin{table}[thbp]
\def~{\hphantom{0}}
\begin{tabular}{ccccc}
 			& \multicolumn{2}{c}{Block weights}  \\
Design method 							& $\bullet$  & $\times$ & Bayes efficiency & Time (processor-seconds)\\
Likelihood, na\"ive outcome-enumeration 			& 0$\cdot$744 & 0$\cdot$256 & 100$\cdot$00 & 1$\cdot$65{}$\times$$10^7$ \\ 
Likelihood, interpolated outcome-enumeration			& 0$\cdot$749 & 0$\cdot$251 & 99$\cdot$96 &  5$\cdot$19{}$\times$$10^6$\\
Adjusted marginal quasi-likelihood 			& 0$\cdot$748 & 0$\cdot$252 & 99$\cdot$79 & 1$\cdot$80{}$\times$$10^5$\\
Adjusted estimating equations			 	& 0$\cdot$466 & 0$\cdot$534 &97$\cdot$94 & 2$\cdot$20{}$\times$$10^5$
\end{tabular}
\caption{Example 2: details of Bayesian designs \label{tab:lookup-example-designs}.
%The blocks corresponding to these weights can be identified in \mbox{Fig. \ref{fig:lookup-bayes-effs}} from the plotting symbols $\bullet$ and $\times$.
Above, $\bullet$ and $\times$ correspond to symbols in \mbox{Fig. \ref{fig:lookup-bayes-effs}}(a)--(d).
The Bayes efficiency of  $\xi$ is $\exp[\{\psi(\xi) - \sup_{\xi'}\psi(\xi')\}/p]$.}
\end{table}

% Lookup
%
%\begin{table}[bthp]
%\def~{\hphantom{0}}
%\tbl{Bayesian designs for the example of Section \ref{sec:EX1} \label{tab:lookup-example-designs} \textbf{REVISE}}{
%\begin{tabular}{ccccc}
%Approximation &Block ($i$)& $x^\T_{i1}$ & $x^\T_{i2}$ & Weight ($w_i$) \\[1ex]
% Computational& 1 & (1$\cdot$000, -0$\cdot$402)& (-1$\cdot$000, 0$\cdot$402) & 0$\cdot$370 \\ 
% & 2 & (-0$\cdot$144, 1$\cdot$000) & ( 0$\cdot$144, -1$\cdot$000) & 0$\cdot$227 \\ 
% & 3 & (1$\cdot$000, -1$\cdot$000) & (-1$\cdot$000, 1$\cdot$000) & 0$\cdot$402 \\[1ex] 
%Attenuation-adjusted &1 & ( 1$\cdot$000, -0$\cdot$386) & (-1$\cdot$000, 0$\cdot$386)& 0.364 \\
%  & 2 & (-0.128, 1.000) & (0$\cdot$129, -1$\cdot$000)&0.232\\
%&3& ( 1$\cdot$000, -1$\cdot$000) & (-1$\cdot$000,  1$\cdot$000) & 0.404\\[1ex]
%Marginal quasi-likelihood &1 & (-0$\cdot$874, 1$\cdot$000) & (0$\cdot$874 , -1$\cdot$000) & 0$\cdot$205 \\ 
% & 2 & (-0$\cdot$449, 0$\cdot$049) & (0$\cdot$449 , -0$\cdot$049) & 0$\cdot$328 \\ 
% & 3 & (-1$\cdot$000, 0$\cdot$659) & (1$\cdot$000 , -0$\cdot$659) & 0$\cdot$294 \\ 
% & 4 & ( 0$\cdot$232,-1$\cdot$000) & (-0$\cdot$232 , 1$\cdot$000) & 0$\cdot$173 \\[1ex]
%\end{tabular}}
%\end{table}

\subsection{Example 3: Locally optimal designs, four factors}
\label{sec:4factor}
We investigated locally optimal designs with $x=(x^{(1)},x^{(2)},x^{(3)},x^{(4)})^\T\in [-1,1]^4$, and
\begin{align*}
\nu(x;u,\beta) = &\beta_0 + \beta_1 x^{(1)} + \beta_2 x^{(2)}  + \beta_3 x^{(3)} + \beta_4 x^{(4)} \\
	&\quad + \beta_{12} x^{(1)}x^{(2)} + \beta_{13}x^{(1)}x^{(3)} + \beta_{14} x^{(1)}x^{(4)} + u \,, \quad u \sim N(0,\sigma^2)\,,
	\label{eq:4factor-model}
\end{align*}
with  $\beta_{\text{att}} = (2,3,0,3,0,0,-2,0)^\T$, $(1,2,1,-3,-1,\frac{1}{4},-\frac{1}{2},3)^\T$, $(0,1,1,1,1,\frac{1}{2},\frac{1}{2},\frac{1}{2})^\T$, and $\sigma^2=1,2 ,5$. Designs were found using the na\"ive outcome-enumeration, adjusted marginal quasi-likelihood and adjusted generalized estimating equation ($\rho$ = 0$\cdot$3, 0$\cdot$5, 0$\cdot$6) approximations with 100, 1000 and 1000 random starts respectively. In all cases, the marginal approximations required less computational effort despite the more thorough search, yielding designs with at least 99.5\% efficiency relative to the design from the na\"ive outcome-enumeration approximation.

%%%%%%%%%%%%%%%%%%%%%%%%%%%%%%%%%%%%%%%%%%%%%%%%%%%%%%%%%%%%%%%%%%%%%%%%

%%%%%%%%%%%%%%%%%%%%%%%%%%%%%%%%%%%%%%%%%%%%%%%%%%%%%%%%%%%%%%%%%%%%%%%%

%%%%%%%%%%%%%%%%%%%%%%%%%%%%%%%%%%%%%%%%%%%%%%%%%%%

\section{Poisson response}

\subsection{Approach}
\label{sec:poisapp}

In this section we demonstrate the use of the marginal quasi-likelihood approximation to find $D$-optimal designs for a Poisson model with random intercept. We compare the designs to those of \citet{niaparast}, who investigated design for this model using a direct quasi-likelihood approximation to the information matrix, and also to the designs from the analytical results of \citet{russell-poiss} for the Poisson model with no random effects. The conditional distribution of the response is assumed to be Poisson, with link function $g(\mu)=\log(\mu)$. In the random intercept model, $u \sim N(0,\sigma^2)$ is a scalar, and $\nu(x;u,\beta) = f^\T(x)\beta + u$.

Quasi-likelihood estimation requires a parametric specification of only the marginal mean and variance of the response, and not a full probability model. \citet{niaparast} obtained a covariance matrix for the resulting parameter estimators using the actual marginal mean and variance for the Poisson random intercept model which are analytically tractable. We shall refer to this as the `direct' approach. In general, there are issues with the use of quasi-likelihood for dependent data \citep[][Ch.9]{mccullagh-nelder}; however the above approach could be viewed as an application of generalized estimating equations \citep[][]{liang-zeger} with a working correlation structure calculated from the full probability model.

\subsection{Comparison of designs, $m=3$}
\label{sec:poisex}

Locally $D$-optimal designs for the Poisson random intercept model were computed by numerically optimizing the determinant of the information matrix under the marginal quasi-likelihood and direct, quasi-likelihood, approximations. The linear predictor structure~\eqref{eq:2factor-model} was assumed, with conditional parameter values $(\beta_0,\beta_1,\beta_2)=(3,1,2)$, together with several values for $\sigma^2$.

%\citet{russell09} derived an analytical form for the $D$-optimal approximate design for the first-order Poisson model with no random effects and no blocking (in other words $\sigma^2=0$ and $m=1$). For the model with two explanatory variables, this optimal design has three equally weighted support points; for $\beta=(0,1,2)^{\T}$, the $D$-optimal design on $[-1,1]^2$ assigns equal weight to $(-1,1)^{\T}$, $(1,0)^{\T}$ and $(1,1)^{\T}$. To use this analytical result as a point of reference, a block size of $m=3$ was chosen.

For $\sigma^2=0$, the designs found numerically coincided with those anticipated by the theoretical results of \citet{russell-poiss} for models with no random effects. For $\sigma^2$ = 0$\cdot$01, 0$\cdot$025, 0$\cdot$05, 0$\cdot$1, each of the designs contains a single support block ($b=1$, $w_1=1$) of the form $\zeta(t)=( (1,1)^\T, (-1,1)^\T, (1,t)^\T )$, with $t$ = -0$\cdot$083, -0$\cdot$091, -0$\cdot$095, \mbox{-0$\cdot$096} respectively. The designs from the two methods agree to three decimal places.

We assess the efficiency, for maximum likelihood estimation, of designs resulting from the choice of $t \in [-1,1]$ by using Monte Carlo integration to approximate $M_\beta(\zeta(t); \theta)$ (Section~\ref{sec:complenum}) and nonparametric smoothing to obtain a surrogate, $\tilde{\Psi}(t)$, for $\Psi(t)= |M_\beta(\zeta(t); \theta)|$ \citep[see also][]{muller1995optimal}. Let $t^\star=\max_{t'\in [-1,1]} \tilde{\Psi}(t')$. Figure \ref{fig:poiss-objfn} shows the approximate efficiency in the neighbourhood of the optimal $t$, obtained from  \mbox{$\tilde{\operatorname{eff}}(t) = \{ \tilde{\Psi}(t)/ \tilde{\Psi}(t^\star) \}^{1/p}$}, together with estimates of $\{\Psi(t)/ \tilde{\Psi}(t^\star) \}^{1/p}$ each using $10^5$ Monte Carlo samples. The total processor time for the Monte Carlo computations was approximately 1$\cdot$5$\times 10^5$s, using a sixteen-core 2$\cdot$6 GHz node. The results indicate that, for all values of $\sigma^2$ considered here, both the marginal and direct quasi-likelihood designs have an efficiency around 100\%, and also any choice of $t$ in \mbox{[-0$\cdot$15, 0]} will be very highly efficient.

\begin{figure}[t]
\begin{centering}
\includegraphics[width=\linewidth]{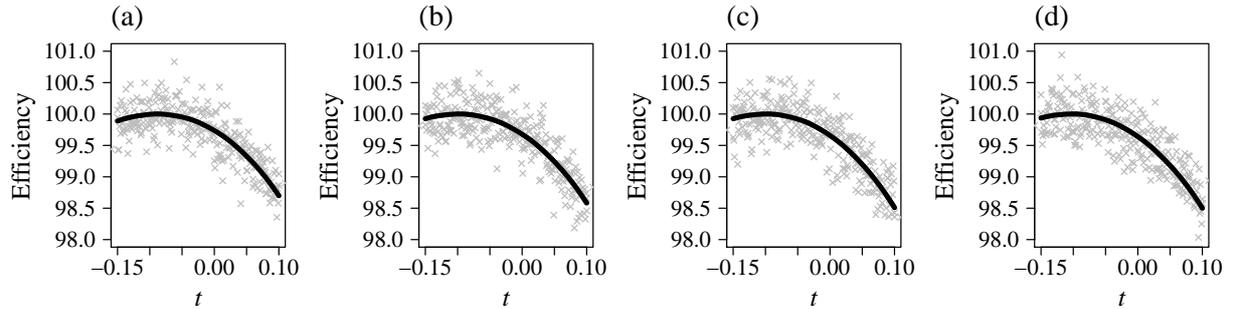}
\end{centering}
\caption{Approximate efficiency of the design $\zeta(t)$ obtained using Monte Carlo approximation and nonparametric smoothing, for the following values of $\sigma^2$: (a) 0$\cdot$01, (b) 0$\cdot$025, (c) 0$\cdot$05, (d) 0$\cdot$1. }
\label{fig:poiss-objfn}
\end{figure}

The direct, quasi-likelihood, approach for the Poisson response is similar to the adjusted marginal or adjusted generalized estimating equations methods for a binary response, in the sense that it accounts for the form of the marginal mean. Theoretically, it has the advantage of not relying on Taylor series approximations. Unlike for binary data, where an unadjusted marginal method is poor, for a Poisson response the unadjusted method has virtually identical performance to the direct method. This is perhaps to be expected if we consider that the normal approximations to the response distribution used in the marginal approximation are much more accurate for Poisson than binary responses. There is essentially no computational advantage to the Taylor-series based approximation and so we would recommend the direct approach as a default first choice.

Note the values of $\sigma^2$ used here are much smaller than those used for binary response models in Section~\ref{sec:bin}; for a Poisson response, $\sigma^2$ is chosen to give a plausible range of marginal overdispersion over $\mathcal{X}$ (for example, approximately 1$\cdot$10--45$\cdot$6 when $\sigma^2=0$$\cdot$$1$), and plausible correlation between responses from units in the same block receiving the same treatment.

%\begin{table}
%\tbl{Locally $D$-optimal designs for the Poisson random intercept model, $m=3$ \label{tab:poiss-designs}}{
%\begin{tabular}{cccc}
%$\sigma^2$&\multicolumn{3}{c}{Design}  \\
% 0$\cdot$01 &(1,1) & (-1,1) & (1, -0$\cdot$083)  \\
% %  &     &      & (1, -0$\cdot$083)& 100.0 (99.3--100.6) \\[0.1ex]
% 0$\cdot$025   &(1,1)      & (-1,1)     & (1, -0$\cdot$091)  \\  
%% &     &      & (1, -0$\cdot$091) & 100.0 (98.8--101.3) \\[0.1ex]    
% 0$\cdot$05 &(1,1) & (-1,1)& (1, -0$\cdot$095)  \\
%   %&    &      & (1, -0$\cdot$095) & 100.0 (97.7--100.2) \\[0.1ex]
%0$\cdot$1  & (1,1)&   (-1,1)   & (1, -0$\cdot$096)   \\
%%&     &      & (1, -0$\cdot$096)& 100.0 (95.5--104.7) \\    
%\end{tabular}}
%\end{table}

%%%%%%%%%%%%%%%%%%%%%%%%%%%%%%%%%%%%%%%%%%%%%%%%%%%%%%%%%%%%%%%%%%%%%%%%
\section{Discussion}

For the logistic random intercept model, use of a correction for the marginal attenuation of the parameters yields much improved designs; in our examples, designs using this idea often performed on a par with those from na\"ive outcome-enumeration. Further investigations, including simulations to assess small sample properties, are available in the first author's Ph.D. thesis. 

\citet{tekle} employed an information matrix approximation derived from penalized quasi-likelihood \citep{bres-clay}. Their approach requires predictions of the random effects, which they approximated at the design stage using Monte Carlo simulation. The resulting approximation is computationally intensive and is not suitable for routine use on more complex problems. Hence, we chose not to pursue this methodology here.

Avenues for future research include developing the necessary methodology to extend the adjusted closed-form approximations to find designs for models with more complex random effects, and extension of the asymptotic results in Section \ref{sec:strongdep} to other link functions for binary response.
%%%%%%%%%%%%%%%%%%%%%%%%%%%%%%%%%%%%%%%%%%%%%%%%%%%%%%%%%%%%%%%%%%%%%%%%%

\section*{Acknowledgements}
The authors thank T.~H.~Waterhouse (Eli Lilly) for helpful discussions. This work was supported by the UK Engineering and Physical Sciences Research Council through a platform grant, a PhD studentship and Doctoral Prize for the first author, and a Fellowship for the second author. It was partly undertaken while the authors were visiting the Isaac Newton Institute for Mathematical Sciences, Cambridge, UK. The authors acknowledge the use of the Iridis computational cluster and associated services at the University of Southampton.

\appendix

\section*{Appendix 1: An algorithm for determining membership of sets $\mathcal{N}(j)$, $\mathcal{Z}(j)$, $\mathcal{P}(j)$}
To obtain an asymptotic approximation that performs reasonably for a broad choice of designs, a decision is required on which $\eta_j$ should be considered `close'; that is, for given $\zeta$ and $j$, which indices should we treat as belonging to $\mathcal{Z}(j)$ in order to apply Theorems 1--3? Below we give the algorithm used in Example 1. The algorithm approximates exponentially decaying error terms as zero.

For calculation of the derivatives, the algorithm iteratively augments $\mathcal{Z}(j)$ with the index, $l$, of the next closest predictor to $\eta_j$ if two conditions are satisfied. Condition (A) concerns the coefficient of $\phi(-\eta_j/\sigma)/\sigma$ in the expression of Theorem~\ref{thm:asymp-approx-deriv}, which is an approximation to an integral of the form $\int_{-\infty}^{\infty} h'(t)f_2(t,\sigma^2) dt$, see equation~(A3) in Appendix 2. The value of this integral decreases as the set $\mathcal{Z}(j)$ is augmented. Condition (B), concerning the same coefficient, is a heuristic that prevents the application of a Taylor approximation when $\Delta_{lj}$ is too large, see (A4) in Appendix 2.

For the probability calculation, we use the expression from part (i) of Theorem \ref{thm:asymp-prob-quasi-incr} unless $\max_{j\in \mathcal{S}_0}\{ \eta_j \}$  and $\min_{j \in \mathcal{S}_1}\{ \eta_j \}$ are close (less than 1 apart), in which case we take $j'=\operatorname{arg\,max}_{\{j\in \mathcal{S}_0 \}}\{\eta_{j}\}$, $l'=\operatorname{arg\,min}_{\{j\in \mathcal{S}_1 \}}\{\eta_{j}\}$, $\mathcal{Z}(j')=\{ j',l' \}$ and use the expression in part (ii) of Theorem \ref{thm:asymp-prob-quasi-incr}. The cutoff distance of $\gamma=1$ is chosen because at this point the probabilities in parts (i) and (ii) should be similar, since $\Phi(-\eta_j/\sigma)-\Phi(-\eta_{l}/\sigma) \approx \frac{\eta_{l}-\eta_{j}}{\sigma}\phi(-\eta_{j}/\sigma)= \frac{1}{\sigma}\phi(-\eta_{j}/\sigma)$.

{\small
\begin{algo}
%\caption{Approximate information matrix using asymptotic formulae}
For each possible outcome $Y$, approximate its contribution, $\frac{1}{P(Y)}\left\{ \frac{\partial P(Y)}{\partial \eta}\right\} \left\{\frac{\partial P(Y)}{\partial \eta}\right\}^\T$,  to the information matrix in \eqref{eq:blockinf} using Theorems 1--3 to approximate $P(Y)$ and $\partial P(Y)/\partial \eta$, and add it to the total.\\[1ex]
To compute $P(Y):$ 
\begin{tabbing}
	\enspace Compute $\lambda_0 = \max_{j\in \mathcal{S}_0} \{ \eta_j \}$ and $\lambda_1 = \min_{j \in \mathcal{S}_1}\{ \eta_j \}$\\
	\enspace If $\lambda_1 \geq \lambda_0 + \gamma$:\\
	\qquad Set $P(Y) \leftarrow \Phi(\lambda_1/\sigma) - \Phi(\lambda_0/\sigma)$ [using Theorem 1(i)]\\
	\enspace If $|\lambda_1- \lambda_0| \leq \gamma$:\\
	\qquad Set $P(Y) \leftarrow \frac{\phi(\lambda_1/\sigma)}{\sigma}$ [using Theorem 1(ii)]\\
	\enspace If $\lambda_1 \leq \lambda_0 - \gamma$, set $1/P(Y) \leftarrow 0$, and do not compute $\partial P(Y)/{\partial \eta_j}$\\
	\qquad i.e. do not include a contribution from this outcome in the information matrix approximation
\end{tabbing}
To compute $\partial P(Y)/{\partial \eta_j}$:
\newcommand{\argmin}{\operatorname*{arg\,min}}
\begin{tabbing}
	\enspace Declare $\mathcal{Z}(j) = \{j \}$\\
	\enspace Set $C_4 = 1$\\
	\enspace Propose augmenting $\mathcal{Z}(j)$ to $\mathcal{Z}'(j) = \{ j, \argmin_{l\neq j} |\eta_j-\eta_l|\}$ \\
	\enspace Iterate until STOP. Given current proposal $\mathcal{Z}'(j)$:\\
	\qquad Calculate $I(j), J(j)$ for $\mathcal{Z}'(j)$, refer to as $I', J'$ respectively \\
	\qquad Set $C'_4 \leftarrow 
						C^{(1)}_{I',J'} + C^{(3)}_{I'-1,J'}
						\sum_{l\in \mathcal{S}_1 \cap \mathcal{Z}'(j)\backslash\{ j\}} \Delta_{lj} 
						- C^{(3)}_{I',J'-1}\sum_{l\in \mathcal{S}_0 \cap \mathcal{Z}'(j)\backslash\{ j\}} \Delta_{lj}$  \\
	\qquad If (A) $0 \leq C'_4 \leq C_4$ and (B) $|C'_4 - C^{(1)}_{I',J'}| \leq |C_4 - C^{(1)}_{I',J'}|$, accept proposal\\
	\qquad\qquad Update $C_4 \leftarrow C'_4$, $\mathcal{Z}(j) \leftarrow \mathcal{Z}'(j)$ \\
	\qquad If did not accept proposal in previous step, then STOP\\
	\qquad Otherwise make new proposal, $\mathcal{Z'}(j) \leftarrow \mathcal{Z}(j) \cup \argmin_{l \not\in \mathcal{Z}(j) } \{ |\eta_l - \eta_{j}| \} $\\
	\enspace Set $\mathcal{N}(j) \leftarrow \{ l : \eta_l < \eta_{l'}, \mbox{ for all } l' \in \mathcal{Z}(j)\}$\\
	\enspace Set  $\mathcal{P}(j) \leftarrow \{ l : \eta_l > \eta_{l'}, \mbox{ for all } l' \in \mathcal{Z}(j)\} $\\
	\enspace If  $\{ \mathcal{S}_1 \cap \mathcal{N}(j)\}  \cup\{ \mathcal{S}_0 \cap \mathcal{P}(j)\} = \emptyset$:\\
	\qquad Deem the outcome as quasi-increasing\\
	\qquad Set $ \frac{\partial P(Y)}{\partial \eta_j} 
					\leftarrow (2y_j -1)\max\left\{ 0,
			\frac{1}{\sigma} \phi\left(\frac{-\eta_j}{\sigma}\right) C_4  
			+ \frac{1}{\sigma^2}
				\phi'\left(\frac{-\eta_j}{\sigma}\right) C^{(2)}_{I(j),J(j)}  \right\}$ [using Theorem 2(i)] \\
	\enspace Else set $\frac{\partial P(Y)}{\partial \eta_j} \leftarrow 0$ [using Theorem 2(ii)]
\end{tabbing}
 \end{algo}}

\section*{Appendix 2: Proofs and further asymptotic results} 

Recall that $\mathcal{Z}(j)= \{ l: \eta_l - \eta_j \to 0 \}$, $\mathcal{N}(j) = \{ l : \eta_l - \eta_j \to -\infty\}$, $\mathcal{P}(j) = \{ l : \eta_l - \eta_j \to \infty \}$, and $\mathcal{S}_0=\{j: y_j=0\}$, $\mathcal{S}_1=\{ j:y_j=1\}$. For the asymptotic results, we require some assumptions repeated here for clarity.\begin{assumption}
$\beta_\text{att} = \beta/\sqrt{1+c^2 \sigma^2}$ is fixed as $\sigma^2 \to \infty$.
\label{assume:beta-att-fixed}
\end{assumption}
\begin{assumption}
For all $j=1,\ldots,m$, either $\eta^\ast_j = f^\T(x_j)\beta_\text{att}$ is fixed or there exists $l\neq j$ with $\eta^\ast_l$ fixed and $\eta^\ast_l - \eta^\ast_j = o(\sigma^{-1})$.
\label{assume:eta-star-condns}
\end{assumption} % Z, N, P, A_j, B_j, conditions on eta^\ast_j (formulated as assumption/condition A1)
\begin{assumption}
There exists $A_j>0$ such that $|\eta_l -\eta_j| > \sigma A_j$ for $l \in \{\mathcal{S}_0\cap \mathcal{N}(j)\} \cup \{\mathcal{S}_1 \cap \mathcal{P}(j)\}$, and $B_j>0$ such that $|\eta_l - \eta_j| > \sigma B_j$ for all $l \in \{ \mathcal{S}_1 \cap \mathcal{N}(j) \}\cup\{  \mathcal{S}_0 \cap \mathcal{P}(j)  \}$. 
\end{assumption}
Define 
\begin{align*}
f_{1,j}(t,\sigma^2) &=  \prod_{l \in \mathcal{S}_1 \cap \mathcal{N}(j)} 
			h(\eta_l - \eta_j + t)
		\prod_{l \in \mathcal{S}_0  \cap \mathcal{P}(j)} 
			\{ 1 -h(\eta_l - \eta_j + t) \}\\
f_{2,j}(t,\sigma^2) &=  \prod_{l \in \mathcal{S}_1 \cap \mathcal{Z}(j) \backslash \{j\}}
			h(\eta_l - \eta_j + t)
				\prod_{l \in \mathcal{S}_0 \cap \mathcal{Z}(j) \backslash\{j\}}
			\{ 1 -h(\eta_l - \eta_j + t) \} \\
f_{3,j}(t,\sigma^2) & =  \prod_{l \in \mathcal{S}_1 \cap \mathcal{P}(j) }
		h(\eta_l - \eta_j + t) 	\prod_{l \in \mathcal{S}_0 \cap  \mathcal{N}(j) }
		\{ 1 -h(\eta_l - \eta_j + t) \}		
\end{align*}
We will mostly suppress the dependence of these functions on $j$ and write $f_1$, $f_2$, $f_3$ where the context is clear. 
Fix $t\in\mathbb{R}$. 
%%If $ \mathcal{S}_1 \cap \mathcal{N}(j) \backslash \{j\}$ is non-empty, then  $\prod_{l \in \mathcal{S}_1 \cap \mathcal{N}(j) \backslash \{j\}}
% 			h(\eta_l - \eta_j + t) \to 0$ as $\sigma^2\to \infty$. Similarly, if $ \mathcal{S}_0 \cap \mathcal{P}(j) \backslash \{j\}$ is non-empty, $\prod_{l \in \mathcal{S}_0  \cap \mathcal{P}(j)\backslash\{j\}}
%		\{ 1 -h(\eta_l - \eta_j + t) \} \to 0$ as $\sigma^2\to \infty$. 
If  $\{ \mathcal{S}_1 \cap \mathcal{N}(j)\} \cup \{ \mathcal{S}_0 \cap \mathcal{P}(j)  \}\neq \emptyset$, then $f_1(t,\sigma^2)\to 0$ as $\sigma^2\to \infty$, otherwise $f_1(t,\sigma^2)=1$. We always have $f_3(t,\sigma^2)\to 1$. We make use of the following lemma.

\begin{lemma}
Suppose $f_3$ is as defined above, and $f_4(t,\sigma^2)$ is measurable as a function of $t$ for all fixed $\sigma^2$, with $0 \leq f_4(t,\sigma^2)\leq K$ for all $t,\sigma^2$, for some $K>0$. Then:  

(i) For any $\epsilon>0$, as $\sigma^2\to \infty$,
\begin{align*}
\int_{-\infty}^{\infty}
	 h'(t) f_3(t,\sigma^2) f_4(t,\sigma^2) dt &= 
\int_{-\infty}^{\infty} 
	h'(t) f_4(t,\sigma^2) dt + O(e^{-\sigma A_j /[(1+\epsilon)m]})  \,,
%\text{(ii) }\quad
%\int_{-\infty}^{\infty}
%	 \phi(t) f_3(t,\sigma^2) f_4(t,\sigma^2) dt &= 
%\int_{-\infty}^{\infty} 
%	\phi(t) f_4(t,\sigma^2) dt + O(e^{-\sigma A_j /[(1+\epsilon)m]}) \,, 
\end{align*}
i.e. replacing $f_3$ by 1 in the integrand incurs only an exponentially decaying error. 

(ii) Suppose that $\Delta_1,\Delta_2$ vary with $\sigma^2$, but $|\Delta_1|, |\Delta_2|\leq \Delta_{\operatorname{max}}$. Then, as $\sigma^2\to \infty$, for any $\epsilon>0$,
\begin{align*}
&\int_{-\infty}^{\infty}
	 h(t+\Delta_1)\{1-h(t+\Delta_2)\}
	  f_3(t,\sigma^2) f_4(t,\sigma^2) dt \\&= 
\int_{-\infty}^{\infty} 
	 h(t+\Delta_1)\{1-h(t+\Delta_2)\}	 
	 f_4(t,\sigma^2) dt + O(e^{-\sigma A_j /[(1+\epsilon)m]})  \,,
\end{align*}
i.e. the integrator, $h'(t)$, in (i) can be replaced by $h(t+\Delta_1)\{1-h(t+\Delta_2)\}$.
\end{lemma}
 The key idea in the proof of Lemma 1 is to approximate the logistic function by a step function. Observe that if $h$ is the logistic function and $S(t) = \mathbb{I}(t>0)$, then there is $L>0$ such that $|h(t)-S(t)| \leq Le^{-|t|}$. Moreover, we can reduce the rate constant for the exponential and still have an upper bound. Thus, given $\epsilon>0$, $|h(t)-S(t)| \leq Le^{-|t|/(1+\epsilon)}$. 
 
 As a prelude to the proof of Lemma 1, we demonstrate exponential convergence of a relatively simple integral to zero. The full proof is more intricate, but does not involve many more ideas. Observe
 \begin{align*}
\left|
\int_{-\infty}^{\infty}
 [
 h(t+\sigma) - S(t+\sigma)] h'(t) dt
\right|
&\leq 
\int_{t+\sigma >0 }
	L e^{-\sigma/(1+\epsilon)-t/(1+\epsilon)} h'(t) dt
+ \int_{t+\sigma <0}
	h'(t)dt	\\
&\leq L e^{-\sigma/(1+\epsilon)} \int_{-\infty}^{\infty} e^{-t/(1+\epsilon)}h'(t) dt
	+ h(-\sigma)\\
	& = O(e^{-\sigma/(1+\epsilon)}) \,.
\end{align*}
 Key to the conclusion is the observation that the integral in the second line is finite. This is true since in the upper and lower tails the integrand is bounded, respectively, by $\lambda e^{-|t|\{1+1/(1+\epsilon)\}}$ and $\lambda e^{-|t|\{1-1/(1+\epsilon)\}}$, where $\lambda>1$. The integral is not finite if $\epsilon=0$.
 \begin{proof}[of Lemma 1] 
Part (i): Observe that
 \begin{align*}
f_3(t,\sigma^2)
		&= \prod_{l \in \mathcal{S}_0\cap \mathcal{N}(j)}
					 h(-(\eta_l-\eta_j+t))
					\prod_{l \in \mathcal{S}_1 \cap \mathcal{P}(j)}
					h(\eta_l -\eta_j +t)  \,.
\end{align*}
Assume $A_j\sigma + t >0$. Then, for $l \in \mathcal{S}_1 \cap \mathcal{P}(j)$, there is a constant $L_m>0$ such that
\begin{align*}
|h(\eta_l - \eta_j + t)- 1| &= 
|h(\eta_l - \eta_j + t) - S(\eta_l - \eta_j + t)|  \leq L_m e^{-|\eta_l - \eta_j + t|/((1+\epsilon)m)}\\
 & \leq L_m e^{-\sigma A_j /((1+\epsilon)m)-t/((1+\epsilon)m)}  \leq L_m e^{-\sigma A_j/((1+\epsilon)m) + |t|/((1+\epsilon)m)} \,.
\end{align*}
By a similar argument, $L_m$ can also be chosen such that, in addition, for $l \in \mathcal{S}_0\cap \mathcal{N}(j)$ and $t < A_j \sigma$, 
\begin{align*}
|h(-(\eta_l-\eta_j + t))-1| \leq L_m e^{-A_j\sigma/((1+\epsilon)m) + |t|/((1+\epsilon)m)} \,.
\end{align*}
Thus, for $-A_j\sigma < t <A_j \sigma$, 
\[
	 \prod_{l \in \{\mathcal{S}_0 \cap \mathcal{N}(j) \} \cup \{ \mathcal{S}_1 \cap \mathcal{P}(j) \}} 
	 \left\{ 
	 1- L_m e^{-\sigma A_j /((1+\epsilon)m)+|t|/((1+\epsilon)m)}
	 \right\} 
	 \leq f_3(t,\sigma^2) \leq 
		1\,.
\]
Let $\kappa = |\{ \mathcal{S}_0 \cap \mathcal{N}(j) \} \cup \{ \mathcal{S}_1\cap \mathcal{P}(j)\}|$, noting $\kappa \leq m$. Binomial expansion of the product yields a conservative bound, 
\begin{align*}
\left|\int_{-A_j\sigma}^{A_j\sigma} h'(t) \{ 1-f_3(t,\sigma^2) \} dt\right| &\leq
\sum_{l=1}^{\kappa} 
{ {\kappa}\choose{l} } L_m^l e^{-l\sigma A_j /((1+\epsilon)m)}  
\int_{-\infty}^{\infty} e^{ l|t|/((1+\epsilon)m)} h'(t)dt\\
&= O(e^{-\sigma A_j/((1+\epsilon)m)}) \,,
\end{align*}
as the integral on the right hand side is finite for $l\leq m$. Moreover, 
$$\left|\int_{A_j\sigma}^{\infty}  h'(t) \{ 1-f_3(t,\sigma^2) \} dt\right| \leq 1- h(A_j\sigma) = O(e^{-A_j\sigma})\,,$$ 
and similarly  
$$\left|\int_{-\infty}^{-A_j\sigma}  h'(t) \{ 1-f_3(t,\sigma^2) \} dt\right| \leq h(-A_j\sigma) = O(e^{-A_j\sigma})\,.$$ 
Overall, 
\[
\left| \int_{-\infty}^{\infty} h'(t) \{ 1- f_3(t,\sigma^2) \} dt \right| = O(e^{-\sigma A_j/((1+\epsilon)m)}) \,.
\]
When combined with the assumption $0\leq f_4(t,\sigma^2)\leq K$, this is adequate to prove the lemma.

Part (ii): First note that there exists $K'>0$ such that, for all $\Delta_1, \Delta_2$ with $|\Delta_1|, |\Delta_2| \leq \Delta_{\operatorname{max}}$,
\begin{equation}
h(t+\Delta_1)\{1-h(t+\Delta_2)\}  \leq K' h'(t) \,.
\label{eq:lem-deltas-Kdash}
\end{equation}
Now consider the case $f_4=1$, for which we have
\begin{align*}
& \int_{-\infty}^{\infty}
h (t+\Delta_1) \{ 1- h(t+\Delta_2) \}
 \{ 1- f_3(t,\sigma^2) \} dt\\
  &\leq 
 K'\int_{-\infty}^{\infty}
 h'(t) \{
 	1-f_3(t,\sigma^2)
 \} dt
=O(e^{-\sigma A_j/((1+\epsilon)m)}) \,,
\end{align*}
as established in part (i). The result for general $f_4$ holds via a similar argument to part (i). 

It can be seen that a conservative choice in \eqref{eq:lem-deltas-Kdash} above is $K' = 4\exp \Delta_{\operatorname{max}}$. To show this, note
$$
R := \frac{h(t+\Delta_1) \{1- h(t+\Delta_2)\} }{h'(t)} 
= \frac{e^{\Delta_1} (1+e^t)^2}{(1+e^{t+\Delta_1})(1+e^{t+\Delta_2})}
= \frac{e^{\Delta_1} (e^{-t}+1)^2 }{ (e^{-t}+e^{\Delta_1} ) (e^{-t} + e^{\Delta_2})} \,.
$$
For $t\geq 0$, use the final expression above to see that $R\leq e^{\Delta_1}\times 4 / e^{\Delta_1+\Delta_2} = 4e^{-\Delta_2} \leq 4e^{\Delta_{\max}}$. For $t<0$, considering the penultimate expression above we see $R \leq e^{\Delta_1}\times 4/1 \leq 4e^{\Delta_{\max}}$.
\end{proof}
 \begin{proof}[of Theorem 2 (Derivatives)]
  
%Let $\mathcal{N}(j) = \{ l : \eta_l-\eta_j \to -\infty \}$, $\mathcal{Z}(j) = \{ l : \eta_l - \eta_j \to 0 \}$, and $\mathcal{P}(j)=\{ l : \eta_l - \eta_j \to \infty\}$. Given a general (block) outcome $y \in \{0,1\}^m$, let $\mathcal{S}_1 = \{ l : y_l=1\}$, $\mathcal{S}_0 = \{ l : y_l=0\}$. Then
Part (i): The derivative is given by
\begin{align}
\frac{\partial P(Y)}
{\partial \eta_j} 
&= 
(2y_j-1)
\int_{-\infty}^{\infty}
 h'(\eta_j+\sigma u)
	 \prod_{l \in \mathcal{S}_1\backslash \{j\}}
 			h(\eta_l + \sigma u)
	\prod_{l \in \mathcal{S}_0 \backslash\{j\}}
		\{ 1 -h(\eta_l + \sigma u) \}
		\phi( u ) du \notag \\
		&= \frac{ (2y_j-1)}{\sigma}
\int_{-\infty}^{\infty}
 h'(t) 
f_1(t,\sigma^2)f_2(t,\sigma^2)f_3(t,\sigma^2) 	\phi\left( \frac{t}{\sigma} - \frac{\eta_j}{\sigma}\right) dt \,,
	\label{eq:general-derivInt}
%
%
%		&=\frac{ (2y_j-1)}{\sigma}
%\int_{-\infty}^{\infty}
% h'(t) 
%	 \prod_{l \in \mathcal{S}_1 \cap \mathcal{N}(j) \backslash \{j\}}
% 			h(\eta_l - \eta_j + t) \notag \\
%			& \quad \quad \quad \quad \quad
%		 \prod_{l \in \mathcal{S}_1 \cap \mathcal{Z}(j) \backslash \{j\}}
% 			h(\eta_l - \eta_j + t) 	
%		 \prod_{l \in \mathcal{S}_1 \cap \mathcal{P}(j) \backslash \{j\}}
% 			h(\eta_l - \eta_j + t) 	\notag \\
%			& \quad \quad \quad \quad \quad \quad
%	\prod_{l \in \mathcal{S}_0 \cap  \mathcal{N}(j)  \backslash\{j\}}
%		\{ 1 -h(\eta_l - \eta_j + t) \}
%	\prod_{l \in \mathcal{S}_0 \cap \mathcal{Z}(j) \backslash\{j\}}
%		\{ 1 -h(\eta_l - \eta_j + t) \} \notag \\
%		& \quad \quad \quad \quad \quad \quad
%	\prod_{l \in \mathcal{S}_0  \cap \mathcal{P}(j)\backslash\{j\}}
%		\{ 1 -h(\eta_l - \eta_j + t) \} \,\,
%		\phi\left( \frac{t}{\sigma} - \frac{\eta_j}{\sigma}\right) dt
%
\end{align}
since, from Assumption \ref{assume:eta-star-condns}, $\mathcal{N}(j)\cup \mathcal{Z}(j)\cup \mathcal{P}(j)=\{1,\ldots, m\}$.
%where the second line follows by substitution, and the fact that by Assumption \ref{assume:eta-star-condns}, $\mathcal{N}(j)\cup \mathcal{Z}(j)\cup \mathcal{P}(j)=\{1,\ldots, m\}$.
%		$ \prod_{l \in \mathcal{S}_1 \cap \mathcal{P}(j) \backslash \{j\}}
% 			h(\eta_l - \eta_j + t) \to 1$ and $\prod_{l \in \mathcal{S}_0 \cap  \mathcal{N}(j)  \backslash\{j\}}
%		\{ 1 -h(\eta_l - \eta_j + t) \} \to 1$. 
If $ \mathcal{S}_1 \cap \mathcal{N}(j)= \mathcal{S}_0 \cap \mathcal{P}(j) =\emptyset$, then $f_1=1$ and, from Lemma 1(i), \eqref{eq:general-derivInt} is equal to
\[
		\frac{ (2y_j-1)}{\sigma} \int_{-\infty}^{\infty}		
			h'(t) f_2(t,\sigma^2) \phi\left( \frac{t}{\sigma} - \frac{\eta_j}{\sigma}\right) dt  + 
			O\left(\frac{1}{\sigma}e^{-\sigma A_j/[(1+\epsilon)m]}\right)\,.
\]
Applying Taylor's theorem (to the normal density), we find an approximation correct to $O(\sigma^{-3})$:
\begin{align}
\frac{\partial P(Y)}
{\partial \eta_j} %&=
%\frac{ (2y_j-1)}{\sigma} 
%\int_{-\infty}^{\infty}
%h'(t) f_2(t,\sigma^2) \left\{
%	\phi(-\eta_j/\sigma)
%	+ (t/\sigma) \phi'(-\eta_j/\sigma)
%	+ (t/\sigma)^2 \phi''(-\tilde{\eta}) 
%\right\}dt + O\left(\frac{1}{\sigma}e^{-\sigma A_j/[(1+\epsilon)m]}\right) \\
&= \frac{ (2y_j-1)}{\sigma}  
\left\{ \phi(-\eta_j/\sigma) \int_{-\infty}^{\infty}
h'(t) f_2(t,\sigma^2) dt 
 + \frac{\phi'(-\eta_j/\sigma)}{\sigma}
 \int_{-\infty}^{\infty}
t h'(t) f_2(t,\sigma^2)  dt \right. \Bigg\} \notag \\
& \quad\quad\quad\quad\quad\quad\quad\quad +O(\sigma^{-3}) 
%\int_{-\infty}^{\infty}
%t^2 h'(t) f_2(t,\sigma^2) \phi''(-\tilde{\eta}) dt \Bigg\}
+O\left(\frac{1}{\sigma}e^{-\sigma A_j/[(1+\epsilon)m]}\right)\,.
\label{eq:derivative-asymp-sigma-third-order}
\end{align}
We now expand $f_2$ in terms of $\Delta_{lj} = \eta_l-\eta_j$ to find a computationally simpler expansion. Recall that $I(j)=|\mathcal{S}_1\cap \mathcal{Z}(j)\backslash \{j \}|$, $J(j)=|\mathcal{S}_0 \cap \mathcal{Z}(j) \backslash \{ j \}|$, and note that
\begin{align}
f_2(t,\sigma^2) %&= \prod_{l\in \mathcal{S}_1 \cap \mathcal{Z}(j)\backslash\{ j\}}
			%	h(\eta_l-\eta_j + t)
			%	\prod_{l\in \mathcal{S}_0 \cap \mathcal{Z}(j)\backslash\{ j\}}
			%	\{ 1 - h(\eta_l -\eta_j + t) \}\\
		&= h(t)^{I(j)}\{1-h(t)\}^{J(j)}
			+ \sum_{l\in \mathcal{S}_1 \cap \mathcal{Z}(j)\backslash\{ j\}}
				\Delta_{lj} h'(t)  h(t)^{I(j)-1} \{ 1- h(t) \}^{J(j)} \notag
				\\ & \quad\quad\quad\quad\quad\quad\quad
				- \sum_{l\in \mathcal{S}_0 \cap \mathcal{Z}(j)\backslash\{ j\}}
				\Delta_{lj} h'(t) h(t)^{I(j)}  \{ 1- h(t) \}^{J(j)-1}
			+ \sum_{l\in \mathcal{Z}(j)}O(\Delta_{lj}^2) \,.
\label{eq:derivative-f2-expand-delta}
\end{align}
Substituting \eqref{eq:derivative-f2-expand-delta} into \eqref{eq:derivative-asymp-sigma-third-order} gives the result.

Part (ii). Applying a similar argument to that in the proof of Lemma 1, to the function $f_1$ in the case $\{ \mathcal{S}_1 \cap \mathcal{N}(j) \}  \cup \{\mathcal{S}_0 \cap \mathcal{P}(j) \} \neq \emptyset$, shows that
 \[
 \int_{-\infty}^{\infty} h'(t) f_1(t,\sigma^2) f_4(t,\sigma^2) dt = O(e^{-\sigma B_j/(1+\epsilon)}) \,.
 \] %\textbf{(Written out in full somewhere in my personal notes) }
Applying this to \eqref{eq:general-derivInt} above gives the result.
\end{proof}
 \begin{lemma}
 \label{lemma:suppmat-lem2}
 Let $f_5 (\eta, u)$ be a function, measurable as a function of $u$ for fixed $\eta$, satisfying \mbox{$0 \leq f_5(\eta,u) \leq K$}. Then:
\[
\int_{-\infty}^{\infty}  f_5(\eta,\sigma u) h(\eta+\sigma u) \phi(u) du = \int_{-\infty}^{\infty}  f_5(\eta,\sigma u) S(\eta+\sigma u) \phi(u) du + O(\sigma^{-1}) \,,
\]
where $S(t)=\mathbb{I}(t>0)$.
 \end{lemma}
 \begin{proof}[of Lemma 2]
 Note that
 \begin{align*}
 \left|
  \int_{-\infty}^{\infty} f_5(\eta,\sigma u)\left[ h(\eta+\sigma u) - S(\eta +\sigma u) \right] \phi (u) du \right|
&\leq K
	\frac{1}{\sigma} \int_{-\infty}^{\infty}
		 \left| h(t) - S(t) \right| \phi(t/\sigma - \eta/\sigma) dt
	 \\
&\leq \frac{ K \phi(-\eta/\sigma)}{\sigma} \int_{-\infty}^{\infty}|D(t)| dt + O(\sigma^{-2}) \,,
 \end{align*}
 where $D(t) = h(t)-S(t)$, by application of Taylor's theorem.
 \end{proof}
 \begin{proof}[of Theorem 1 (Probabilities)] Part (i): Observe
 \begin{align*}
P(Y) &= 
\int_{-\infty}^{\infty} \prod_{j \in \mathcal{S}_1} 
	h(\eta_j + \sigma u)
\prod_{j \in \mathcal{S}_0 } 
	\{1- h(\eta_j + \sigma u)\}
	\phi(u)du\\
&= \int_{-\infty}^{\infty} \prod_{j \in \mathcal{S}_1} 
	\mathbb{I}(\eta_j +\sigma u >0 )
\prod_{j \in \mathcal{S}_0 } 
	\mathbb{I}(\eta_j +\sigma u <0)
	\phi(u)du + O(\sigma^{-1}) \\
&= \int_{-\infty}^{\infty} 
		\mathbb{I}(
				\max_{j\in \mathcal{S}_0}\{\eta_j/\sigma \}
				< -u < 
			\min_{j\in \mathcal{S}_1}\{\eta_j/\sigma \}  ) 
		\phi(u) du + O(\sigma^{-1})\\
&= \max\{ 0, \Phi(\min_{j\in \mathcal{S}_1}\{\eta_j/\sigma \}) - \Phi( \max_{j\in \mathcal{S}_0} \{\eta_j/\sigma \} )\} + O(\sigma^{-1}) \,,
\end{align*}
where the second line follows by repeated application of Lemma \ref{lemma:suppmat-lem2}.\\[0.5ex]
Part (ii): %Note that \[
%P(Y) = 
%\int_{-\infty}^{\infty} 
%	\prod_{j\in \mathcal{S}_0}
%		\{ 1 - h(\eta_j + \sigma u) \}
%	\prod_{j \in \mathcal{S}_1}
%		h(\eta_j + \sigma u) \phi(u) du \,.
%\]
By assumption, there exists $j'\in\mathcal{S}$ such that $\{\mathcal{S}_0 \cap \mathcal{P}(j')\} \cup \{\mathcal{S}_1 \cap \mathcal{N}(j')\} = \emptyset$, $|\mathcal{S}_0 \cap \mathcal{Z}(j')|\geq 1 $ and $|\mathcal{S}_1 \cap \mathcal{Z}(j')|\geq 1$. Thus, taking $l_1 \in \mathcal{S}_1\cap \mathcal{Z}(j')$, $l_2 \in \mathcal{S}_0 \cap \mathcal{Z}(j')$,
\begin{align*}
P(Y) &=
\frac{1}{\sigma}
\int_{-\infty}^{\infty}\Big[
 h(\Delta_{l_1 j'}+t)
\{ 1- h(\Delta_{l_2 j'}+ t) \} \\
 & \qquad\qquad\qquad
 \prod_{l\in\mathcal{S}_0\cap \mathcal{Z}(j') \backslash\{l_2\}}
 	\{ 1- h(\Delta_{lj'}+ t) \}
 \prod_{l \in\mathcal{S}_1\cap\mathcal{Z}(j') \backslash\{ l_1\}}
 h(\Delta_{lj'}+t) \\
 &\qquad\qquad\qquad\qquad\qquad
  f_{3,j'}(t,\sigma^2)
 \phi(t/\sigma - \eta_{j'}/\sigma)
 \Big]dt  \,.
\end{align*}
Since $\Delta_{l_1 j'}, \Delta_{l_2,j'}\to 0$, we have that $\Delta_{l_1 j'}$, $\Delta_{l_2 j'}$ are bounded. Thus, from Lemma 1(ii),
\begin{align*}
P(Y) &=
\frac{1}{\sigma}
\int_{-\infty}^{\infty}
 \prod_{l\in\mathcal{S}_0\cap \mathcal{Z}(j') }
 	\{ 1- h(\Delta_{lj'}+ t) \}
 \prod_{l \in\mathcal{S}_1\cap\mathcal{Z}(j') }
 h(\Delta_{lj'}+t) 
 \phi(t/\sigma - \eta_{j'}/\sigma)
 dt \\
 & \qquad\qquad\qquad+ O\left(\frac{1}{\sigma}e^{-\sigma A_{j'} /[(1+\epsilon)m]}\right) \,.
\end{align*}
This can be approximated using a Taylor expansion in $\Delta_{lj'}$ as
\begin{align*}
P(Y) &=
\frac{1}{\sigma}
\int_{-\infty}^{\infty}
 \prod_{l\in\mathcal{S}_0\cap \mathcal{Z}(j') }
 	\{ 1- h(t) \}
 \prod_{l \in\mathcal{S}_1\cap\mathcal{Z}(j') }
 h(t) 
 \phi(t/\sigma-\eta_{j'}/\sigma )
 dt \\
 & \qquad\qquad\qquad+ \sum_{l \in \mathcal{Z}(j')}O(\Delta_{lj'}/\sigma)+O\left(\frac{1}{\sigma}e^{-\sigma A_{j'} /[(1+\epsilon)m]}\right) \,.
\end{align*}
A formal argument using the mean value form of Taylor's theorem can be made to verify that the additional error incurred by the last step is indeed $\sum_{l\in\mathcal{Z}(j')}O(\Delta_{lj'}/\sigma)$. Applying Taylor's theorem to the normal density function yields
\begin{align}
P(Y) &=
\frac{1}{\sigma}
\int_{-\infty}^{\infty}
 \prod_{l\in\mathcal{S}_0\cap \mathcal{Z}(j') }
 	\{ 1- h(t) \}
 \prod_{l \in\mathcal{S}_1\cap\mathcal{Z}(j') }
 h(t) 
  \phi(- \eta_{j'}/\sigma) 
 dt \notag \\& \qquad \qquad \qquad + 
 \frac{1}{\sigma^2}
\int_{-\infty}^{\infty}
 \prod_{l\in\mathcal{S}_0\cap \mathcal{Z}(j') }
 	\{ 1- h(t) \}
 \prod_{l \in\mathcal{S}_1\cap\mathcal{Z}(j') }
 h(t) 
t \phi'(-\tilde{\eta}_t/\sigma) dt
 \notag \\
 & \qquad\qquad\qquad+ \sum_{l \in \mathcal{Z}(j')} O(\Delta_{lj'}/\sigma) + O\left(\frac{1}{\sigma}e^{-\sigma A_{j'}/[(1+\epsilon)m]}\right) \,,
 \notag%\label{eq:proof-th1ii-taylor}
 \end{align}
 with $\tilde{\eta}_t$ between $\eta_{j'}$ and $\eta_{j'}-t$. Since $h'(t)=h(t)\{1-h(t)\}$, the second integral has the form $\int_{-\infty}^{\infty} h'(t) f_4(t,\sigma^2) dt$, with $f_4$ bounded, and so the overall remainder term is $O(\sigma^{-2})$.
 \end{proof}

\begin{proof}[ of Theorem 3] We show that, for all outcomes, 
\[
\left|\frac{\partial P(Y)}{\partial \eta_j}\right| / P(Y) \leq 2 \,.
\]
%In combination with Theorems 1 and 2, this is adequate to prove the result.
 Observe that both \mbox{$h(t),1-h(t) \geq (1/2)e^{-|t|}$} and $h'(t) \leq e^{-|t|}$. For $j \in \mathcal{S}_1$,
\begin{align}
P(Y) &= 
\int_{-\infty}^{\infty}
h(\eta_j + \sigma u)
 \prod_{l \in \mathcal{S}_1\backslash \{j\}} 
	h(\eta_l + \sigma u)
\prod_{ l \in \mathcal{S}_0\backslash \{j\} } 
	\{1- h(\eta_l+ \sigma u)\}
	\phi(u)du \notag \\
&   \geq (1/2)\int_{-\infty}^{\infty}
 e^{-|\eta_j+\sigma u|}
  \prod_{l \in \mathcal{S}_1\backslash \{j\}} 
	h(\eta_l + \sigma u)
\prod_{ l \in \mathcal{S}_0\backslash \{j\} } 
	\{1- h(\eta_l+ \sigma u)\}
	\phi(u)du \,, \notag%\label{eq:pf-th-a2-ii-prob}
\end{align}
and the same lower bound holds for $j\in\mathcal{S}_0$.
Compare with the derivative,
\begin{align}
\left|\frac{\partial P(Y)}{\partial \eta_j}\right|& = \int_{-\infty}^{\infty}
h'(\eta_j +\sigma u)  \prod_{l \in \mathcal{S}_1\backslash \{j\}} 
	h(\eta_l + \sigma u)
\prod_{ l \in \mathcal{S}_0\backslash \{j\} } 
	\{1- h(\eta_l+ \sigma u)\}
	\phi(u)du \notag \\
	&\leq \int_{-\infty}^{\infty}
e^{-|\eta_j +\sigma u|} \prod_{l \in \mathcal{S}_1\backslash \{j\}} 
	h(\eta_l + \sigma u)
\prod_{ l \in \mathcal{S}_0\backslash \{j\} } 
	\{1- h(\eta_l+ \sigma u)\}
	\phi(u)du \,. \notag%\label{eq:pf-th-a2-ii-deriv}
\end{align}
Thus $\left|\frac{\partial P(Y)}{\partial \eta_j}\right|/P(Y) \leq 2$ and, in conjunction with Theorems 1 and 2, the theorem is proved.

%can be seen via dividing \eqref{eq:pf-th-a2-ii-deriv}  by \eqref{eq:pf-th-a2-ii-prob}. The integrals in these two inequalities are identical and so cancel.
\end{proof} 

Propositions 1 and 2 below give additional details of the behaviour of the random intercept logistic regression model for large $\sigma^2$.
\begin{proposition}
As $\sigma^2 \to \infty$
(i) the probability that the outcome in any given block is increasing is $1+O(\sigma^{-1})$; (ii) the probability that the outcomes in all blocks are increasing is $1+O(\sigma^{-1})$.
\end{proposition}

\begin{proof} Consider a single block. Without loss of generality, we may assume the units in the block are ordered such that $\eta_1 \leq \ldots \leq \eta_m$. We define $\eta_0 = -\infty$, $\eta_{m+1} = \infty$ for convenience.
Then, the increasing outcomes are $(00\ldots0)$, $(00\ldots01)$,  $(00\ldots11)$, $\ldots$, $(11\ldots1)$. From Theorem 1, with $Y=(y_1,\ldots,y_m)^T \in \{0,1\}^m$, a within block outcome vector,
$$
P\{Y \text{ is increasing and first 1 occurs at }y_j \}
= \Phi( - \eta_{j-1}/\sigma) - \Phi(-\eta_j /\sigma) + O(\sigma^{-1})\,.
$$
Overall,
\begin{align*}
P\{ Y \text{ is increasing} \} &= \sum_{j=1}^{m+1} P\{Y \text{ is increasing and first 1 occurs at }j\text{th position} \} \\
&=
\sum_{j=1}^{m+1} [ \Phi( - \eta_{j-1}/\sigma) - \Phi(-\eta_j /\sigma)] + O(\sigma^{-1}) \\
&= \Phi(-\eta_0 /\sigma) - \Phi(-\eta_1/\sigma) + \Phi(-\eta_1/\sigma) -\Phi(-\eta_2/\sigma) \\
& \quad\quad+ \ldots -\Phi(-\eta_m/\sigma) + \Phi(-\eta_m/\sigma) - \Phi(-\eta_{m+1}/\sigma) + O(\sigma^{-1}) \\
& = \Phi(\infty) - \Phi(-\infty) + O(\sigma^{-1})\\
&= 1 + O(\sigma^{-1}) \,.
\end{align*}
By independence of blocks, the probability that the outcomes of all blocks are increasing is $(1+O(\sigma^{-1}))^n$. This equals $1 + nO(\sigma^{-1}) + O(\sigma^{-2}) = 1+O(\sigma^{-1})$, by binomial expansion.
\end{proof}
\begin{proposition} For any $\sigma>0$, if all blocks have increasing outcomes, then the parameters of the logistic model with fixed block effects and linear predictor
$$
\eta_{ij} = f^\T(x_{ij})\beta + \gamma_i \,,\quad i=1,\ldots,n; j=1,\ldots,m\,,
$$
are not estimable by maximum likelihood.
\end{proposition}

\begin{proof}
The argument is essentially the same as for separation in the standard logistic model case. From the assumptions that the outcomes in each block are increasing, for each $i$ there exists $\tilde{\eta}_i$ such that
\begin{align*}
f^\T(x_{ij})\beta > \tilde{\eta}_i \iff y_{ij}=1 \\
f^\T(x_{ij})\beta < \tilde{\eta}_i \iff y_{ij}=0
\end{align*}
For $\lambda>0$, consider $\theta_\lambda= (\beta_\lambda,\gamma_\lambda) = (\lambda \beta, -\lambda \tilde{\eta})$. Let $\delta_{ij} = f^T(x_{ij})\beta - \tilde{\eta}_i$, and note that $\delta_{ij} > 0$ if $y_{ij}=1$ and $\delta_{ij}<0$ if $y_{ij}=0$. Then 
\begin{align*}
\Pr(  y | \theta_\lambda) = 
\prod_{i,j \,:\, y_{ij}=1} h(\lambda\delta_{ij})
\prod_{i,j \,:\, y_{ij}=0}
\{1 - h (\lambda \delta_{ij}) \}
\end{align*}
As $\lambda \to \infty$, $\Pr(y|\hat{\theta}_\lambda) \to 1$. Thus, given any set of finite parameter values (which must have likelihood less than 1), there is a $\hat{\theta}_\lambda$ that has higher likelihood. Thus there is no set of finite parameter values that maximize the likelihood.
\end{proof}

\section*{Appendix 3: Estimation of parameters for large $\sigma^2$} 
To assess the difficulty of estimating the fixed parameters for varying $\sigma$, for parameters $\beta_i\neq 0$ we examined the approximate relative error of estimation,
$$
\operatorname{sd}(\hat{\beta}_i)/|\beta_i|  \approx [M^{-1}_\beta(\xi^\ast ; \theta)]^{1/2}_{ii} / (\sqrt{n}|\beta_i|) \,,
$$
with the optimal design for each of the parameter combinations in Section 5.2. For $\beta_i=0$, we compared the standard deviation of $\hat{\beta}_i$ to the magnitude of the smallest nonzero parameter,
$$
\operatorname{sd}(\hat{\beta_i})/\min_{\{i : \beta_i \neq 0\}} |\beta_i| \,.
$$ 
These relative errors are plotted in Figure A1 above, with each colour corresponding to a different parameter scenario. 
We use relative errors as these are most appropriate when comparing estimation quality for parameter values of potentially quite different sizes. 

\begin{figure}
\begin{center}
\includegraphics[width=\linewidth]{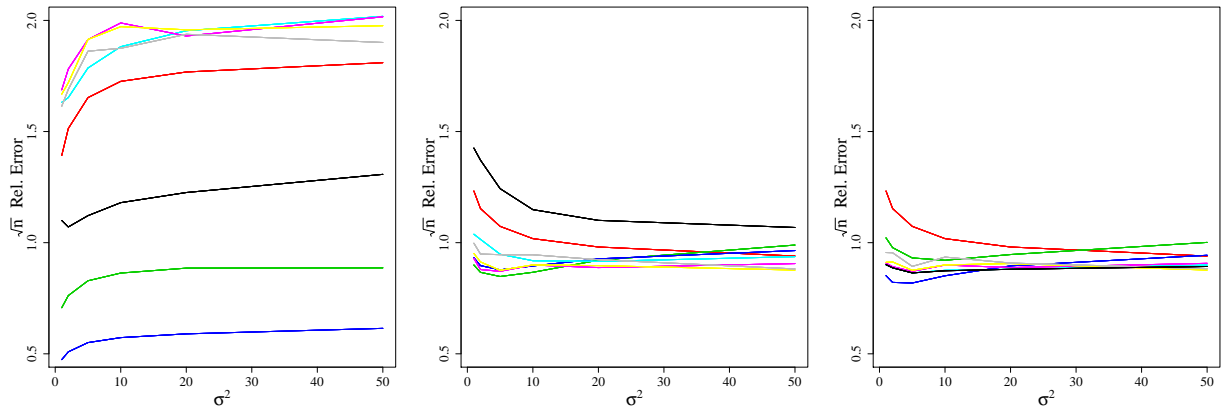}
\caption{Sample-size normalized approximate relative estimation error (for $\beta_i\neq 0 $, $[M^{-1}_\beta(\xi^\ast ; \theta)]^{1/2}_{ii} / \beta_i$), for varying $\sigma^2$. Relative errors on the same coloured line correspond to parameter scenarios with the same values of the marginal parameters. The first, second, and third panels correspond to $\beta_0$, $\beta_1$, and $\beta_2$, respectively.}
\end{center}
\end{figure}

We see that, for comparable values of the marginal parameters, the relative errors for $\beta_1$ and $\beta_2$ tend to decrease or remain approximately the same as $\sigma$ increases. For these parameters, therefore, the same level of estimation precision may be achieved for large $\sigma$ with no additional experimental units or, in some cases, up to 40\% fewer units. The relative error for $\beta_0$ increases with $\sigma$ by 18--30\% in our examples. Hence for the largest $\sigma$, 39--69\% more experimental units are needed to maintain the same level of estimation precision in $\beta_0$. However, $\beta_0$ is often the parameter of least interest. These sample size considerations make clear that useful experimentation remains possible for large $\sigma$ though, of course, detailed results for particular applications may vary.

\bibliographystyle{agsm}
\bibliography{../biom-refs.bib}
\vspace{-5cm}
\end{document}